\documentclass[12pt]{article}
\usepackage{bibentry}
\usepackage{amsthm}
\usepackage{pdfsync}
\usepackage{color}
\usepackage{multirow}
\usepackage{graphicx}
\usepackage{amsmath, amssymb, amsfonts, amsthm}
\usepackage{epstopdf}
\usepackage{graphics}
\usepackage{geometry}
\usepackage{xcolor}
\usepackage{setspace}
\usepackage{setspace}
\usepackage{caption}
\usepackage{subcaption}
 \usepackage{tikz}
 \usepackage{float}
\tikzstyle{densely dashed}=[dash pattern=on 4pt off 2pt]

\newcommand{\field}[1]{\mathbb{#1}}
\newcommand{\R}{\field{R}}

\newcommand{\N}{\field{N}}
\newcommand{\Z}{\field{Z}}

\newcommand{\E}{\field{E}}
\newcommand{\pp}{\mathcal{P}}
\newcommand{\FF}{\mathcal{F}}

\newcommand{\bs}{\boldsymbol}

\def\argmin{\mathop{\mbox{argmin}}}

\theoremstyle{Conjecture}
\theoremstyle{example}
\theoremstyle{remark}
\theoremstyle{lemma}
\theoremstyle{definition}
\theoremstyle{corol}
\theoremstyle{proposition}
\theoremstyle{condition}
\newtheorem{theorem}{Theorem}[section]
\newtheorem{example}{Example}[section]
\newtheorem{remark}{Remark}[section]

\newtheorem{proposition}{Proposition}[section]

\newtheorem{corollary}{Corollary}[section]

\def\lf{\lfloor}
\def\rf{\rfloor}

\usepackage{booktabs}
\usepackage[authoryear]{natbib}
\begin{document}
  \title{Prediction in locally stationary time series}
\author{\small Holger Dette \\
\small Ruhr-Universit\"at Bochum \\
\small Fakult\"at f\"ur Mathematik \\
\small 44780 Bochum \\
\small Germany \\
\and
\small Weich Wu \\
\small  Tsinghua University\\
\small  Center for Statistics\\
\small  Department of Industrial Engineering\\
\small 10084 Beijing China
}

\maketitle
\begin{abstract}
We develop an estimator for the high-dimensional covariance matrix  of a locally stationary process
with a  smoothly varying   trend and use this statistic to derive consistent predictors
in non-stationary time series. In contrast to the currently available methods for this problem the predictor
developed here does not rely on fitting an autoregressive model and  does not require a vanishing trend.  The finite sample properties of the new methodology are illustrated by means of a simulation study and a financial indices study.
\end{abstract}




\section{Introduction}
\label{sec1}
\def\theequation{1.\arabic{equation}}
\setcounter{equation}{0}
	
An important problem in time series analysis is to predict or forecast  future observations  from  a given a stretch of data, say $X_{1}, \ldots , X_{n}$, and numerous authors
have worked on this problem.   Meanwhile
there is a  well developed theory for  prediction under the assumption
of stationary processes
[see for example \cite{brockwell2002introduction}, \cite{bickel2011banded}, \cite{mcmurry2015high} among  many  others]. On the other hand,    if  data is obtained over a long stretch of time it may be unrealistic to assume that the stochastic structure of  a time series  is stable. Moreover, in many shorter time series non-stationarity can also be observed and prediction under the assumption of stationarity  might be misleading.

A common approach to
deal with this problem of   non-stationarity is to assume
a location scale model with a  smoothly changing  trend  and variance but a stationary error process, say $X_n = \mu (n) + \sigma (n) \varepsilon_n$ [see, for example, \cite{van2004forecasting}, \cite{stuaricua2005nonstationarities}, \cite{zhao2008confidence},  \cite{guillaumin2017analysis}, \cite{das2017predictive}].   In this case  the  trend and variance function can be estimated and prediction can be
performed applying methods for stationary data to the standardized  residuals.
However, there  appear also more sophisticated
features of non-stationarity  in the data, which are not captured by a  a simple  location scale model, such as time-changing kurtosis or skewness,
and the standardized  residuals obtained by this procedure  may not be stationary.

To address this type of     non-stationarity various mathematical concepts  modeling a  slowly-changing stochastic structure have been developed in the literature [see for example,  \cite{priestley1988non}, \cite{dahlhaus1997fitting}, \cite{nason2000wavelet}, \cite{zhou2009local} or \cite{VO2012}]. The corresponding stochastic processes
are usually called {\it locally stationary}
and the problem of predicting future observations in these models is  a very challenging one.
An early reference is \cite{fryzlewicz2003forecasting} who  considered centered {\it  locally stationary wavelet processes}.
In this model the sample covariance matrix
in  the  prediction equation is not estimable and   the authors proposed an approximation
 using  the (uniquely defined) wavelet spectrum.
  \cite{van2004forecasting} considered the prediction problem
  in a location scale model with  a smoothly changing variance and stationary error process.
 More recent work on forecasting in centered locally stationary time series can be found in
 \cite{roueff2018prediction} and  \cite{kley2019predictive}.
The first named authors  investigated a predictor based on  auto-regression   of a given order,
while  \cite{kley2019predictive} considered  predictors in  stationary and locally stationary models for (possibly) non-stationary data and selected the ``better''  prediction among the two estimates.
A common feature of  most of these methods is that they are all based on auto-regressive fitting.

In the present paper we contribute to this literature and propose an alternative  method for prediction in
physically dependent
locally stationary  times series, which does not rely on auto-regressive fitting and is therefore more flexible. To be precise
we consider the model
\begin{align} \label{1.1}
X_{i,n}=\mu(i/n)+\epsilon_{i,n},~~i=1, \ldots , n
\end{align}
where $\mu$ is a deterministic and smooth mean or trend function on the interval $[0,1]$ and $\{ \epsilon_{i,n} : i =1, \ldots , n \}_{n\in \N} $ is a triangular array
modelled by a  locally stationary process in the sense of \cite{zhou2009local} - see Section \ref{sec2} for mathematical details.
We then estimate the regression function $\mu$ by local linear smoothing and define a banded
estimator for the corresponding   auto-covariance matrix
 \begin{align}
 \label{1.2}\Sigma_{n}= \big
 \{{\rm Cov} (X_{i,n},X_{j,n} ) \big \}_{1\leq i,j\leq n}
 \end{align}
 from the residuals of the nonparametric fit, where the width of the band increases with the sample size.
 Banded estimates of auto-covariance matrices have been considered by \cite{wu2009banding} and  \cite{mcmurry2010banded}  for {\bf centered}  and {\bf stationary}
processes using the fact that in this case  the matrix
$\Sigma_{n}$ in \eqref{1.2} is a Toeplitz matrix.
Neither of these results is applicable under the assumption of non-stationarity (even if the locally stationary process $\{X_{i,n}\}_{i=1,\ldots,n}$ in \eqref{1.1} is centered).

In Section \ref{sec3} we establish consistency (with respect to  the operator norm) of the new covariance operator for locally stationary processes with a time varying mean function.
These results are then used in Section \ref{sec4} to develop new prediction methods, which  - in contrast to the currently available literature -  do not use
autoregressive fitting.  In Section \ref{sec5} we investigate  the finite sample properties of the estimator of the covariance matrix and
compare the new predictor with the currently available methodology.  Finally, all proofs of our main theoretical results and technical details can be found in   Section \ref{sec6}.

\section{Locally stationary processes}
\label{sec2}
\def\theequation{2.\arabic{equation}}
\setcounter{equation}{0}

Consider the time series model \eqref{1.1} where $\{  \epsilon_{i,n}: i=1,\ldots,n\}_{n \in \mathbb{N}}$ is an array of centered random variables, and $\mu:[0,1]\rightarrow \mathbb R$ is a smooth mean function.  More precisely we assume
\begin{description}
	\item (M1)  The function $\mu$
	in model \eqref{1.1} has a  Lipschitz continuous second order derivative    on the interval $[0,1]$.
\end{description}

 In order to model a local stationary error  process
  we use a concept introduced by  \cite{zhou2009local}.
To be precise, define for an $L^q$-integrable random variable $X$ its norm by  $\| X \|_q=(\mathbb{E} [|X|^q])^{1/q} (q \geq 1)$,   let
$\{\varepsilon_i : i \in \mathbb{Z}  \}$ denote a sequence of independent identically distributed observations and define
  $\FF_i = ( \ldots ,\varepsilon_{i-2},\varepsilon_{i-1},\varepsilon_i )$.
We assume that there exists  a function
 $G : [0,1]\times \mathbb R^\N \rightarrow \mathbb R$   such that
  \begin{align} \label{hol7}
  \epsilon_{i,n}=G(i/n, \FF_i)
  \end{align}
is a well defined random variable.   For arbitrary functions $G$ it is not guaranteed that the stochastic structure  of $\{\epsilon_{i,n} \colon i \in \mathbb Z\} $ varies smoothly, but we can achieve this
by the following assumptions.
\begin{description}	\item(L1)
For some $q\geq 2$ we have that
$$
\sup_{t\in[0,1]}\|G(t,\FF_0)\|_q<\infty.
$$
	\item(L2) The function $G$ is differentiable with respect to the first coordinate and there exists a constant $M>0$ such that for all
 $t,s\in[0,1]$
 $$
 \Big \|\frac{\partial}{\partial t}G(t,\FF_0)-\frac{\partial}{\partial t}G(s,\FF_0) \Big \|_2\leq M | t-s | .
 $$
\end{description}
Next we quantify  the dependence structure. For this purpose  let $\{\varepsilon_i' : i \in \mathbb{Z} \}   $ denote  an independent  copy of $\{\varepsilon_i : i \in \mathbb{Z}\}$,
define $\FF_i^* =(\ldots ,\varepsilon_{-2},\varepsilon_{-1},\varepsilon_0',\varepsilon_{1}, \ldots ,\varepsilon_i)$
and
	 \begin{align*}
	\delta_q(G,i)=\sup_{t\in [0,1]}\|G(t,\FF_i)-G(t,\FF^*_i)\|_q
	\end{align*}
	as a measure of dependence. We assume for  the  same $q\geq 2$ as in assumption (L1) that
	
\begin{description}	
	\item(L3)  
	There exists a constant $\chi\in(0,1)$ such that
	$$
\delta_q(G,i)=O(\chi^{i}).
	$$
\end{description}
\begin{example} \label{ex1} {\rm
A prominent example of   this non-stationary model  is a locally stationary $AR(p$) process
where the filter in \eqref{hol7} is defined by
   	\begin{align} \label{LocAR}
   	G(t,\FF_i)=\sum_{s=1}^pa_s(t)G(t,\FF_{i-s})+{\sigma}(t)\varepsilon_{i}
   	\end{align}
where $(\varepsilon_i)_{i\in \Z}$ is a sequence of independent identically distributed centered  random
variables with $\|\varepsilon_1\|_q<\infty$, and
 $a_1 , \ldots , a_{p}, \sigma: [0,1] \to \R$,  are    for smooth functions
 such that for some $\delta_0>1$ the polynomial 
  $1-\sum_{s=1}^pa_s(t)z^s$  has  {no} roots in the disc  $\{ z \in \mathbb{C} \colon |z|\leq \delta_0 \} $.   If  the functions $a$ and  {$\sigma$} have bounded derivatives,
  $G(t,\FF_i)$ has a MA representation of the form  $G(t,\FF_i)={\sigma}(t)\sum_{j=0}^\infty c_j(t)\epsilon_{i-j}$, where $c_1, c_2 ,  \ldots  $ are  smooth functions  with derivatives satisfying $|c'_j(t)|\leq M\chi^j$ for $j\geq 0$. Therefore  assumptions  (L1)-(L3) hold for model \eqref{LocAR}.
 It has been shown in \cite{zhou2013inference} that Model \eqref{LocAR} can approximate the time-varying $AR(p)$ model in \cite{dahlhaus1997fitting}.
  }
\end{example}

\begin{remark} \label{stationary} 
{\rm 
Note that the  definition of a locally stationary error process contains the case that
 each row of
 $\{\epsilon_{i,n} \colon i \in \mathbb Z\}_{n\in \N}  $
 is stationary, that  is $G(t, \FF_i ) = H(\FF_i) $ for some function $H : \mathbb R^\N \to  \R$. In this case the random variables
$
 \epsilon_{i,n}=H(\FF_i)
$
 do not depend on $n$, Assumption (L2) is obviously satisfied and  Assumption  (L1) and (L3) reduce to
 	\begin{description}
  \item(S1)
  For some $q\geq 2$,
$\|H(\FF_0)\|_q<\infty$.
  	\item(S2)
  	There exists a constant $\chi\in(0,1)$ such that
 $$
 \delta_q(H,i)=\|H(\FF_i)-H(\FF^*_i)\|_q =O(\chi^{i})~.
  	$$
  \end{description}
  }
  \end{remark} 

If assumption (L1) holds the covariance matrix $\Sigma_n=(\sigma_{i,j,n})_{1\leq i,j\leq n} $ in \eqref{1.2} is well defined, where
\begin{align} \label{2.1}
\sigma_{i,j,n}=\mbox{Cov}(X_{i,n},X_{j,n}) =
\E(G(i/n,\FF_i)G(j/n,\FF_j)) .
\end{align}
Throughout  this paper  we do not reflect the dependence on $n$ in the notation of the entries of a matrix, whenever it is clear from the context.
For example we will use   $\sigma_{i,j}$ instead  {of} $\sigma_{i,j,n}$ and similarly  a simplified notation for corresponding estimates.
We also define the  (time dependent)  auto-covariances
 \begin{align}
\label{gamma}
 \gamma_k(t)=\E(G(t,\FF_i)G(t,\FF_{i+k})) ~~~(
 k\in \Z)
 \end{align}
 of the stationary (for fixed $t \in [0,1]$) process $\{G(t,\FF_i)\}_{i \in \Z} $.
To estimate the covariances  in \eqref{2.1} 
we use a local linear regression estimate of the function $\gamma_k$.
In order to prove consistency  of this estimator
we require  a smoothness condition on the auto-covariances in \eqref{gamma}, which is formulated as follows.
\begin{description}
	\item (A1) For any $ k \in \Z$ the function $\gamma_k$ in \eqref{gamma} is differentiable with derivative $\dot \gamma_k(t)=\frac{\partial}{\partial t}\gamma_k(t)$.  There exists constants $D_k$ such that  for   all $t,s\in [0,1]$
	\begin{align*} 
     \left |\dot \gamma_k(t)-\dot  \gamma_k(s)\right|\leq D_k|t-s|.
      \end{align*}
	\end{description}
An application of  the   Cauchy-Schwarz inequality and the  dominated convergence theorem show  that  a sufficient condition for assumptions  (L2) and (A1), is given by (L1) and
\begin{align*} 
\sup_{t\in [0,1]}\Big \|\frac{\partial^2 }{\partial t^2}G(t,\FF_0)\Big \|_2<\infty.
\end{align*}
In the following section we will  use the  local linear estimates for the function $\gamma_k$ to  define a  banded estimate of the covariance  matrix $\Sigma_n$ of a locally  stationary process of the form \eqref{1.1}  and investigate its asymptotic  properties for increasing sample size. We also discuss a corresponding  estimator in the stationary case because usually estimators are  studied under the assumption of a centered stationary process, that is  $\mu \equiv 0$. In  the subsequent Section \ref{sec4}  we use these results  for prediction in  locally stationary processes with a non-vanishing trend.

\section{Covariance matrix estimation}
\label{sec3}
\def\theequation{3.\arabic{equation}}
\setcounter{equation}{0}

The estimation of the covariance matrix   has attracted considerable attention in the literature.  We refer among many others to the work of
\cite{bickel2008covariance}, \cite{bickel2008regularized}
for high-dimensional independent identically distributed data
and  \cite{anderson2003introduction}, \cite{wu2009banding}, \cite{chen2013covariance}, \cite{box2015time}, and \cite{mcmurry2015high} who considered this problem for time series.
Most authors consider the  case of  a vanishing trend, i.e.   $\mu \equiv 0$,  and assume  that  the error process  $\{\epsilon_{i,n} : i=1,\ldots,n \}$ is   a sequence of independent identical  observations or a     stationary series.
For example, in the case of a stationary centered process  \cite{wu2009banding}  proposed the  banded estimator
\begin{align}\label{estimator-banding}
\tilde  \Sigma_n= \{\tilde \sigma_{i,j}\mathbf 1(|i-j|\leq l_n), 1\leq i,j\leq n\}
	\end{align}
of the matrix $\Sigma_n$,
	where
	$\mathbf 1(A) $ denotes the indicator function of the set $A$ and $$
	\tilde \sigma_{i,j}=\frac{1}{n-|i-j|}\sum_{s=1}^{n-|i-j|}X_{s,n}X_{s+|i-j|,n},
	$$
	is the sample auto-covariance of $\{X_{1,n}, \ldots , X_{n,n} \}$ at lag $|i-j|$ 	and $l_n \in \N$ denotes a tuning parameter  satisfying $l_n\rightarrow \infty$, $l_n=o(n)$ as $n \to \infty$.
		\cite{mcmurry2010banded} modified this statistic such that the new  estimator leaves the  band  intact, and then gradually down-weighs increasingly distant off-diagonal entries instead of setting them to zero as in the banded matrix case. Both  estimators  use the fact that for stationary processes the matrix  $\Sigma_n$  is  a Toeplitz matrix.
	
 Note that the  estimator \eqref{estimator-banding} 
is not consistent for the  auto-covariance if the mean function is not constant. As there   are many applications where time  series have  a smoothly  changing mean function  we  begin our discussion  analyzing a mean-corrected estimator of the matrix $\Sigma_n$ for a   stationary {error} process of the form \eqref{1.1}, which avoids this problem.

Let  $\hat \mu $  be   the local linear estimator
 defined by
\begin{align}\label{locallinear}
(\hat \mu(t),\hat {\dot{\mu}}(t))^\top=\argmin_{(\beta_0,\beta_1)\in \mathbb R^2}\sum_{i=1}^n \big(X_{i,n}-\beta_0-\beta_1(i/n-t)\big)^2K\Big(\frac{i/n-t}{\tau_n}\Big)
\end{align}
where  $\tau_n$ denotes the bandwidth. For the kernel $K$  we make the following assumption:
\begin{description}	\item(K)
The kernel $K$ is a   symmetric,   continuously differentiable, bounded density function supported on the interval  $[-1, 1]$.
\end{description}
We consider   the residuals
 \begin{align} \label{res}
 \hat \epsilon_{i,n}=X_{i,n}-\hat \mu(i/n)
 \end{align}
 obtained from the local linear fit and denote by
$$
\hat\sigma^\dag_{i,j}
=\frac{1}{n-|i-j|}\sum_{s=1}^{n-|i-j|}\hat \epsilon_{s,n}\hat \epsilon_{s+|i-j|,n} ~~~(i,j=1 , \ldots , n)
$$
the sample auto-covariance of the residuals $\{\hat \epsilon_{1,n} , \ldots , \hat \epsilon_{n,n} \}$ at lag $|i-j|$.
Finally, we define  for $ l_{n}  \in \N $ the banded matrix
\begin{align}
\label{Mean-Corrected}
\hat \Sigma^\dag_{n}= \{\hat\sigma^\dag_{i,j}\mathbf 1(|i-j|\leq l_n)\},
\end{align}
as an estimator of the matrix $\Sigma_n$. It will  be shown below  that the estimator
$\hat \Sigma^\dag_{n}$ is consistent for $\Sigma_n$
in the case of a strictly stationary error process.
To measure the distance between two matrices (of increasing dimension) we introduce  the operator norm
\begin{align*}
\rho(A)=\max_{x\in \mathbb R^n:|x|=1}|Ax|
\end{align*}
 of a matrix $A$, where $|\cdot|$ denotes the Euclidean norm
 (note  that  $\rho^2(A)$ is the largest eigenvalue of the matrix $A^\top A$).

\begin{theorem}\label{thm1}
Assume that $n\tau_n^6=o(1)$, $n\tau_n^3\rightarrow \infty$, $l_n\rightarrow \infty$, $\frac{l_n^2}{n}=o(1)$. If conditions (K), (S1), (S2) and (M1) hold, then
	\begin{align*}
	\|\rho(\hat\Sigma^\dag_{n} -\Sigma_n)\|_{q/2} = O (r^\diamond_n) ,
		\end{align*}
		where the sequence $r^\diamond_n$ is defined by
	$$
		r^\diamond_n=l_n(\tau_n^2+(n\tau_n)^{-1/2})+\frac{l_n^2}{n}+\chi^{l_n}.
		$$
	\end{theorem}

 Theorem \ref{thm1} establishes consistency of the estimator of the covariance matrix in model \eqref{1.1} 
 in the operator norm under the assumption of   a stationary error process. However,
there also exist  many time series exhibiting  a non-stationary behaviour in the higher order moments and dependence structure  [see \cite{stuaricua2005nonstationarities}, \cite{elsner2008increasing}, \cite{guillaumin2017analysis}
among others], and
estimation under the assumption of a location model with a stationary error process  might be misleading.
In this case the  estimator $\hat \Sigma^\dag_{n}$ in \eqref{Mean-Corrected}
is not necessarily consistent since the unknown covariance matrix $\Sigma_n$ is not a Toeplitz  matrix.
To address this problem we  propose an alternative
approach which also  yields a  consistent estimator for non-stationary time series.
Roughly speaking, we estimate the elements $\sigma_{i,j} $ in the matrix $\Sigma_n$ by
\begin{align} \label{varest}
\hat  \sigma_{i,j} =
 \hat   \gamma_{|i-j|}\Big (\frac{i+j}{2n}\Big ),
\end{align}
where $\hat   \gamma_k (t) $ is a local linear estimate of
 the auto-covariance function \eqref{gamma} of the process $\{G(t,\FF_i)\}_{i \in \Z} $.

To be precise,  we distinguish between  a  lag of odd or even order and define
 \begin{align}\label{May5-35}
(\hat \gamma_{k}(t),\hat \gamma'_{k}(t))^\top=\argmin_{(\beta_0,\beta_1)\in \mathbb R^2}\sum_{i=1}^n\big(\hat \epsilon_{i-k/2,n}\hat \epsilon_{i+k/2,n}-\beta_0-\beta_1(i/n-t)\big)^2K\Big (\frac{i/n-t}{b_n}\Big )
\end{align}
if the  lag  $k$ is of   even order, where $b_n$ is a bandwidth  and
the   residuals $\hat  \epsilon_{i,n} $ are defined in \eqref{res}. In \eqref{May5-35} we use the notation $\hat \epsilon_{i,n}=0$ if the index $i$ satisfies $i<0$ or $i>n$. Similarly, for an  odd lag  $k$ we  define
\begin{align}\label{May5-36}
\hat \gamma_k(t)=\frac{1}{2}\big(\hat \gamma^+_k(t)+\hat \gamma^-_k(t)\big),
\end{align}
where
\begin{align*}
(\hat \gamma^+_{k}(t),(\hat \gamma_{k}^+)'(t))^\top=\argmin_{(\beta_0,\beta_1)\in \mathbb R^2}\sum_{i=1}^n\big(\hat \epsilon_{i-(k-1)/2,n}\hat \epsilon_{i+(k+1)/2,n}-\beta_0-\beta_1(i/n-t)\big)^2K\Big (\frac{i/n-t}{b_n}\Big ),\\
(\hat \gamma^-_{k}(t),(\hat \gamma_{k}^-)'(t))^\top=\argmin_{(\beta_0,\beta_1)\in \mathbb R^2}\sum_{i=1}^n\big(\hat \epsilon_{i-(k+1)/2,n}\hat \epsilon_{i+(k-1)/2,n}-\beta_0-\beta_1(i/n-t)\big)^2K\Big (\frac{i/n-t}{b_n}\Big ).
\end{align*}
The estimator of the element $\sigma_{i,j}$ in $\Sigma_n$
is finally defined by \eqref{varest}  and for  the covariance
matrix we use again a banded estimator, that
is
\begin{align} \label{locest}
\hat \Sigma_{n}:=\Big (\hat \gamma_{|i-j|} \big (\frac{i+j}{2n})\mathbf 1(|i-j|\leq l_n \big ) \Big )_{1\leq i,j\leq n}.
\end{align}
Our next result yields the consistency of this estimator in the operator norm.

\begin{theorem}\label{thm2} Assume that $n\tau_n^3\rightarrow \infty$, $n\tau_n^6=o(1)$, $\frac{l_n^2}{n}=o(1)$, $l_nb_n^2=o(1),$ 
$$l_n((nb_n)^{-1/2}b_n^{-2/q}+\tau_n^2+(n\tau_n)^{-1/2})=o(1)
~ \mbox{ and } ~~~b_n^2\sum_{k=0}^{l_n}D_k=o(1)
$$ 
If the  conditions (K),  (L1)--(L3), (A1) and (M1)
are satisfied, then
we have 
\begin{align*} 
\|\rho(\hat \Sigma_{n}-\Sigma_n)\|_{q/2}=O( r_n),
\end{align*}
where the sequence $r_n$ is defined by  
\begin{equation}\label{rn}
r_n=l_n((nb_n)^{-1/2}b_n^{-2/q}+\tau_n^2+(n\tau_n)^{-1/2})+\frac{l_n^2}{n}+\chi^{l_n}+ b_n^2 \sum_{k=0}^{l_n}D_k = o(1) .
\end{equation}
\end{theorem}

\begin{remark}
	{\rm ~
	\begin{itemize}
	\item[(a)]
	In the case of a stationary and centered time series it
	has been demonstrated by \cite{mcmurry2015high} that tapering  can improve the performance of simply banded  estimators of the  covariance matrix and  similar arguments apply to the  covariance  estimators  \eqref{Mean-Corrected} and \eqref{locest} proposed in this paper  for  stationary times series with a time varying mean function and for locally stationary times series. To be precise
	consider the situation in Theorem \ref{thm2}  and define the tapering function  (other tapers could be used as well) by
		\begin{align*}
	\kappa(x)=
	(2-|x| )  \mathbf{1}(1\leq |x|\leq 2 ) + 
	  \mathbf{1}(|x|<1)
	\end{align*}
	The tapered and banded  estimate of the covariance matrix $\Sigma_n$ is now defined by
	\begin{align*}
	\hat \Sigma^{tap}_{n}:=\Big (\kappa\Big (\frac{|i-j|}{l_n}\Big )\tilde \gamma_{|i-j|}(\frac{i+j}{2n})\Big)_{1\leq i,j\leq n}.
	\end{align*}
	Using the same arguments  as  in the proof of  Theorem \ref{thm2} it can be shown
	that
		\begin{align}
	\|\rho(\hat \Sigma^{tap}_{n}-\Sigma_n)\|_{q/2}
	= O (r_n) ,\notag
	\end{align}
	where the sequence $r_n$ is defined in \eqref{rn}.
\item[(b)]
It is worthwhile to mention  that recently \cite{ding2018estimation} proposed an alternative estimate of the
 the precision matrix  $\Sigma_{n}^{-1}$ of a  {\bf centered} locally stationary series, which is
  based on a  Cholesky decomposition. In contrast the estimator $\hat \Sigma_{n}^{-1}$
   considers the inverse of  a  banded estimator of the  covariance matrix of a locally stationary series  with
   a smoothly  varying trend.
\end{itemize}   
}
\end{remark}

     \section{Prediction}\label{sec4}
     \def\theequation{4.\arabic{equation}}
\setcounter{equation}{0}

     In this section we discuss some applications of the proposed estimators in the problem to perform predictions 
      in locally stationary processes. For centered time series this problem has been recently investigated by  \cite{roueff2018prediction}, \cite{kley2019predictive} who proposed to fit
  a locally stationary  AR model and perform the  prediction using an AR approximation.  In this section, we suggest
an alternative method which is not based on  AR fitting.
To be precise, assume that   we observe a stretch of data $X_{1,n}, \ldots, X_{m,n}$ from the model \eqref{1.1}
     and that we are interested in a  prediction of the next observation $X_{m+1,n}$.
     To be precise, our
       aim is the   construction of best linear predictor of $X_{m+1,n}$ based on $X_{1,n},\ldots  ,X_{m,n}$.  For this purpose we define
   \begin{align}\label{April-6-hatX}
 X^{\rm Pred}_{m+1,n}:=a_{m+1,n}+\sum_{s=1}^m a_{m+1-s,n}X_{s,n}=\mathbf a_m ^{\top }   \mathbf X_{m,n} ,
   \end{align}
   where   $\mathbf X_{m,n}=(1,X_{1,n},...,X_{m,n})^\top$    and       the prediction vector $\mathbf a_m=(a_{m+1,n},a_{m,n}...,a_{1,n})^\top:=(a_{m+1,n}, (\mathbf{a_m}^{*} )^{\top} )^{\top}$ is given  by
 \begin{align}\label{April 6-33}
   \mathbf a_m=(a_{m+1,n}, (\mathbf{a_m}^{*} )^{\top} )^{\top}= \argmin_{\theta\in \mathbb R^{m+1}}\E(X_{m+1,n}-\theta^\top\mathbf X_{m,n})^2.
   \end{align}
In order to estimate the vector $ \mathbf a_m$ we define  the  local linear estimators
  from the sample   $X_{1,n}, \ldots, X_{m,n}$ by
    \begin{align}\label{localnew}
   (\hat \mu^{1:m} (t),\hat {\dot{\mu}}^{1:m}(t))^\top=\argmin_{(\beta_0,\beta_1)\in \mathbb R^2}\sum_{i=1}^m(X_{i,n}-\beta_0-\beta_1(i/n-t))^2K\left(\frac{i/n-t}{\tau_n}\right) ,
   \end{align}
and denote by
   \begin{equation}\label{hneu1}
   \Sigma_{n,m}=(\sigma_{i,j,n})_{1\leq i,j\leq m} = \big({\rm Cov} (X_{i,n}, X_{j,n})\big)_{1 \leq i,j \leq m}
  \end{equation}
     the covariance matrix of  the vector $(X_{1,n}, \ldots, X_{m,n})^T$.
    The residuals \eqref{res}
   for estimating the auto-covariances are then replaced by residuals
   by 
   $$
    \hat \epsilon_{i,n}^{1:m}
    =X_{i,n}-\hat \mu^{(1:m)} (i/n) \quad (i =1, \ldots , m)
    $$
    from the nonparametric fit from the data $X_{1,n}, \ldots, X_{m,n}$.
   Next, we define $\hat \gamma_k^{1:m}$ as the analogue of the estimator \eqref{May5-35} (if the lag $k$ is even) and \eqref{May5-36} (if the lag is odd), where the residual $\hat \epsilon_{\ell,n}$ is replaced by 
   $\hat \epsilon_{\ell,n}^{1:m}$.
       We further define
   \begin{align} \label{hneu2}
   \hat  \Sigma_{n,m}:=\Big (\hat \gamma^{1:m}_{|u-v|}\big (\frac{u+v}{2n} \big )\mathbf 1(|u-v|\leq l_n)\Big )_{1\leq u,v\leq m}
   \end{align}
  as  a banded  estimator of the  covariance  matrix  $\Sigma_{n,m}:=\text{Cov}(X_{i,n},X_{j,n})_{1\leq j\leq m}$
  in \eqref{hneu1}. It can be shown
  that, if the assumptions   of Theorem \ref{thm2} are satisfied  and   $m\geq \lf cn\rf$ for some positive constant $c$,
  	\begin{align}   \label{rnm}
\|\hat \Sigma_{n,m}-\Sigma_{n,m}\|_{q/2}=O(r_n) ,
  	\end{align}
where the sequence  $r_n$ is defined in  \eqref{rn}.
    We shall construct a predictor based on $ \hat \Sigma_{n,m}^{-1}$ and for  this purpose we show
    that the  consistency of the estimator $\hat \Sigma_{n,m}$   in    \eqref{rnm}
can be transferred to its inverse. \\
 Throughout this paper  we denote $\lambda_{min}(A)$  the minimum eigenvalue of a symmetric matrix $A$ and make the
following  assumption.
\begin{description}
	\item (E1)  There exists a constant $c>0$ such that
	$$
	\eta= \liminf_{n\to \infty  }\inf_{\lf cn\rf\leq m\leq n}\lambda_{min}(\Sigma_{n,m}) >0.
	$$
\end{description}

\begin{corollary}\label{Corol2}Assume that the conditions of Theorem \ref{thm2} and condition (E1) are satisfied. If
$n\to \infty $,  $\lf cn\rf\leq m\leq n$ we have
\begin{align}\label{hatSigmanm}
\rho(\hat \Sigma_{n,m}^{-1}-\Sigma_{n,m}^{-1})= O_{\mathbb{P}}(r_n)
\end{align}
\end{corollary}

  We can now define an estimate $\hat{ \mathbf {a}}_m=(\hat a_{m+1,n},(\hat {\mathbf{a}}_m^*)^\top)^\top$ 
     of the vector $\mathbf a_m$ in \eqref{April 6-33}
by
$$
\hat a_{m+1,n} =\hat \mu^{1:m}(m/n)-\sum_{s=1}^m\hat a_{m+1-s,n}\hat \mu^{1:m}(s/n),\label{hatam}\\
$$
and 
   \begin{eqnarray}\label{hat_a_star}
   \hat {\mathbf a}_m^*  & =&(\hat a_{m,n},...,\hat a_{1,n} )^\top=\hat\Sigma_{n,m}^{-1}\boldsymbol {\hat\gamma}_{n}^{1:m}, 
   \end{eqnarray}
where
  \begin{align}
  \nonumber
         \boldsymbol {\hat\gamma}_{n}^{1:m} 
         &= ({\hat \gamma}_{n,m}^{1:m},{\hat\gamma}_{n,m-1}^{1:m},...,{\hat\gamma}_{n,1}^{1:m})^\top,
\\
     \nonumber 
       {\hat\gamma}_{n,s}^{1:m} & =\hat \gamma^{1:m}_{s}\Big(\frac{2m-s+1}{2n}\Big)\mathbf 1(1\leq s\leq l_n ).
         \end{align}
The final predictor of $X_{m+1,n}$ is defined by
   \begin{align}\label{hol6}
 \hat X^{\rm Pred}_{m+1,n}:=  \hat a_{m+1,n}+\sum_{s=1}^m  \hat a_{m+1-s,n}X_{s,n},
   \end{align}

   \begin{theorem}\label{Thm4}
   	Assume that the conditions of Theorem \ref{thm2} and  assumption (E1) are satisfied,  $\liminf_{n\rightarrow 0} \frac{l_n}{\log n}\geq \eta>0$
	and assume that there
exists a constant $c\in(0,1)$ such that for $m\geq cn$, $m\geq \lf nb_n\rf$.

  (a) The vector $\hat {\mathbf{a}}_m=(\hat a_{m+1,n}, (\hat {\mathbf a}_m^*)^\top)^\top$  is a consistent estimator of the coefficient vector $\mathbf a_m$ of the best linear predictor defined in \eqref{April 6-33}, i.e.,
   	\begin{align*}
   |\hat {\mathbf a}^*_m-\mathbf a^*_m|=O_{\mathbb{P}}(r_n), \quad\hat a_{m+1,n}- a_{m+1,n}=O_{\mathbb{P}}(r^\circ_n)
   	\end{align*}
	where $r_n$ is defined in \eqref{rn}, and 
   	\begin{align} \label{rno}
   		 r_n^\circ =(l_n^{1/2}\log^{1/2} n)r_n+\sqrt n\chi^{l_n}.
   		\end{align}
  (b) Assume that $r_n^\circ=o(1)$. If the error $\epsilon_{i,n}$ is a locally stationary AR($p$) process  as defined in Example \ref{ex1} and
	\begin{description}
   		\item (P1) $n^{\frac{1}{q}}r_n=o(1)$.
   		\item (P2) $\delta_q(\dot G, i)=O(\chi^i)$,
   		\item (P3) $\sup_{t\in [0,1]}\|\dot G(t,\FF_i)\|_q<\infty$,
   	\end{description}
	where $\dot G(t,\FF_i)=\frac{\partial}{\partial t}G(t,\FF_i)$ denotes the derivative of the filter $G$, we have
   	\begin{equation} \label{th3stat}
	\frac{X_{m+1,n}-\hat X^{\rm Pred}_{m+1,n}}{\sigma(\frac{m+1}{n})}\Rightarrow \varepsilon_{1}
   	\end{equation}
   	where $\Rightarrow $ denotes the convergence in distribution and $\varepsilon_1$ denotes the error in model \eqref{LocAR} .
   	\end{theorem}
   
The rate $r_n^\circ$  in  \eqref{rno}  results from  convergence rate of the nonparmetric estimate of  the 
time-varying mean and does not appear if the  trend is not estimated because it  is known to be $0$.
  Conditions (P2) and (P3) can be verified by checking the coefficients of the MA representation of the locally stationary AR process \eqref{LocAR}. They assure that for any $i,j$, the process  $\{\E(G(t,\FF_i)G(s,\FF_j))\}_{t,s\in [0,1]}$ is sufficiently smooth on $[0,1]\times [0,1]$.

      \begin{remark} \label{posdef}
   {\rm
    Similar arguments as given in the proof of  Theorem \ref{thm2} show  that the  estimator $\hat \Sigma_{n,m}$ is   positive definite if the sample size is sufficiently large.
   However, for finite sample sizes the matrix  $\hat \Sigma_{n,m}$ can be singular. As the prediction in   \eqref{hol6}
   requires a  non-singular   sample covariance matrix we propose  in applications to  replace
   the estimator $\hat \Sigma_{n,m}$   by a  a positive definite estimator, say  $\hat \Sigma^{pd}_{n,m}$, which is defined as follows.
If $\hat  \Sigma_{n,m}=U_{n,m}V_{n,m}U_{n,m}^{\top}$   is the  spectral decomposition of $\hat \Sigma_{n,m}$  and $V_{n,m} = \mbox{diag} (v_1,\ldots, v_{m})$ is the  diagonal matrix  containing the
corresponding eigenvalues,
we define 
 \begin{equation}\label{Oct4.8}
  \hat \Sigma^{pd}_{n,m}:=U_{n,m}V^{+}_{n,m}U_{n,m}^{\top}
  \end{equation}
  where   $V^{+}_{n,m}$  is a diagonal matrix with its $i$th diagonal element given by
   $$
   v_{i}^{+} = \max \Big  \{v_i, \frac{10 \int_{0}^{\frac{m}{n}} \hat \gamma^{1:m}_0(t)dt}{m^\beta}\Big \}~,~i=1,\ldots. m
   $$
    for some $\beta>0$.  As a rule of thumb, we choose $\beta=0.5$ because for this choice       $\rho(\hat \Sigma_{n,m}^{pd}-\hat \Sigma_{n,m})=O(n^{-\beta}) =O(r_n)$.
This type of modification  has been also advocated by \cite{mcmurry2010banded} and \cite{mcmurry2015high} for stationary time series.
Using similar argument  as in the   proof of Theorem \ref{thm2} of this paper and in the proof of Theorem 3 of \cite{mcmurry2010banded},
it can be shown that $\|\hat \Sigma^{pd}_{n,m}-\Sigma_{n,m}\|_{q/2}=O(r_n)$.
Now  the arguments given in the proof of Corollary 1 of \cite{wu2009banding}
yield an analogue of Corollary \ref{Corol2}, that is
$$
\rho((\hat \Sigma^{pd}_{n,m})^{-1}-\Sigma_{n,m}^{-1})= O_{\mathbb{P}}(r_n).
$$
A careful  inspection   of the proof of Theorem \ref{Thm4} finally shows that  its assertion remains valid, if $\hat \Sigma_{n,m}$  in \eqref{hatSigmanm} is replaced by $\hat \Sigma^{pd}_{n,m}$.
}
\end{remark}

   	\section{Implementation and numerical results}
   	\label{sec5}
   	\def\theequation{5.\arabic{equation}}
\setcounter{equation}{0}

   	To implement our method we need to choose several tuning parameters: the bandwidths $\tau_n$ and   $b_n$ for the local linear estimators of the trend $\mu$ and auto-covariance
function  $\gamma_k$  and the width  $l_n$ of the banded estimator of the covariance matrix $\Sigma_n$. For choosing $\tau_n$, we recommend the Generalized Cross Validation (GCV) method
proposed  in \cite{zhou2010simultaneous}.

    More precisely, let $\hat \mu^{1:m}(\cdot,\tau), 1\leq i\leq m$ be the local linear estimate of the mean trend defined in \eqref{localnew} using bandwidth $\tau$, then  we choose $\tau_n$ as 
    \begin{align*} 
   	\tau_n=\argmin_\tau \frac{n^{-1}\sum_{i=1}^m ( X_{i,n}-\hat \mu^{1:m}(i/n,\tau) )^2 }{(1-\sum_{i=1}^m(T^{1:m}_{\tau,ii})/n)^2},
   	\end{align*}
	where  $T^{1:m}_{\tau,ii}$ is the $i_{th}$ diagonal entry of the matrix
   	$$
   	J^{1:m}_0 \big ( (X^{1:m}(i/n))^\top W^{1:m}_\tau(i/n)X^{1:m}(i/n) \big )^{-1}(X^{1:m}
   	(i/n))^\top W^{1:m}_\tau(i/n),
   	$$
   	$J^{1:m}_0$ and    $X^{1:m}(i/n)$
  are  	 $m\times 2$  matrices defined by
  \begin{eqnarray*}
  J^{1:m}_0 &=&
  \left( \begin{matrix}
1& 1&  \ldots & 1 \\ 0& 0& \ldots & 0
  \end{matrix} \right)^\top ~,~~
  X^{1:m}(i/n) =
  \left( \begin{matrix}
1& 1&  \ldots & 1 \\ \frac{1-i}n  & \frac{2-i}n &  \ldots & \frac{m-i}{n}
  \end{matrix} \right)^\top,
  \end{eqnarray*}
  respectively, and  $W_\tau(x)$ is an $m \times m$ diagonal matrix with elements
    $\big\{ K \big (\frac{x-s/n}{\tau} \big ) \big \}_{s=1, \ldots  m}$.
The bandwidth $b_n$ for the estimation of the auto-covariance function $\gamma_k$ in \eqref{gamma} is defined similarly.
For example, if $k$ is even,  we choose $b_n$ as
\begin{align} \label{bnselect}
b_n=\argmin_c \frac{n^{-1}\sum_{i=1}^m ( \hat \epsilon^{1:m}_{i-k/2,n}\hat \epsilon^{1:m}_{i+k/2,n}-\hat  \gamma_k^{1:m}(i/n,c) )^2 }{(1-\sum_{i=1}^m(  T^{1:m}_{c,ii})/n)^2},
   	\end{align}
   	where $\hat {\gamma}_k^{1:m}(i/n,c)$ is the local linear estimator with bandwidth $c$ defined  as  in  \eqref{May5-35} using $m$ observations 
	and    $T^{1:m}_{c,ii}$ is defined as  in the previous paragraph.
   	
To motivate  the choice of  the width  $l_n$ in the  banded estimator of the covariance matrix, note that
   	\begin{align}
   	 \sqrt{n}\Big(\frac{1}{n}\sum_{i=1}^{m\wedge (n-k)} \epsilon_{i,n}\epsilon_{i+k,n}-\int_0^{\frac{m}{n}\wedge 1}\gamma_k(t)dt\Big)\Rightarrow {\cal N} (0,\tilde \sigma^2_k),
   	\end{align}
	[see  Section 4.3 in \cite{zhang2012inference}],
   	where $\tilde \sigma_k^2=\int_0^{\frac{m}{n}\wedge \frac{n-k}{n}}g^2(t)dt$, and the function $g^2$ is the long-run variance of the locally stationary
	process $\{\epsilon_{i,n}\epsilon_{i+k,n}\}_{i=1}^{n-k}$. For its estimation we use a statistic proposed by \cite{dette2018change}, which is defined as follows.
Consider the partial sum of lag $k$
$$
^kS^{1:m}_{r_0,r_1}=\sum_{i=r_0}^{r_1}\hat \epsilon^{1:m}_{i,n}\hat \epsilon^{1:m}_{i+k,n},
$$
where we use the notation
$\hat \epsilon^{1:m}_{i,n}=0$ if the index $i$ satisfies $i<1$ or $i>m$. For an  integer $b\geq 2$  we introduce the quantities
$$
^k\Delta^{1:m}_{j,b}=\frac{^k S^{1:m}_{j-b+1,j}- {^k S^{1:m}_{j+1,j+b}}}{b}.
$$
 Finally, we define for $t\in[b/n,(m-b)/n]$
\begin{align*} 
\hat g^2(t) =\sum_{j=1}^n\frac{b {(^k\Delta_{j,b}^{1:m})^2}}{2}\omega(t,j),
\end{align*} 
where
$$
\omega(t,i)=K\Big(\frac{i/n-t}{b_n}\Big){\Big /}\sum_{ {i}=1}^nK\Big(\frac{i/n-t}{b_n}\Big)
$$
and the bandwidth $b_n$ is given by \eqref{bnselect} with $\hat \epsilon^{1:m}_{i-k/2,n}\hat \epsilon^{1:m}_{i+k/2,n}$ there replaced by $\hat \epsilon^{1:m}_{i,n}\hat \epsilon^{1:m}_{i+k,n}$.
For $t\in[0,b/n)$ and   $t\in((m-b)/n,m/n]$ we define
$\hat g^2 (t)=\hat g^2(b/n)$  and
$\hat g^2 (t)=\hat g^2 ((m-b)/n)$, respectively.
	 Finally, we propose
	\begin{align}\label{criteria-March-2019}
   	 l_n=\max \Big\{l\in [l_0,l_1] ~\Big | ~ n^{-1/2}|\sum_{i=1}^n\hat \epsilon^{1:m}_{i,n}\hat \epsilon^{1:m}_{i+l,n}|\geq \kappa({0.01}) \hat \sigma_l\Big\},
   	\end{align} 	
	as a data-driven choice of the width $l_n$, where  $ \kappa(\alpha)$ is the
			  $\frac{1+(1-\alpha)^{1/(l_1-l_0+1)}}{2}$-quantile of the standard normal distribution
	and  $l_{0}$ and $l_1$ are  constants (if the set  $\big\{ n^{-1/2}|\sum_{i=1}^n\hat \epsilon^{1:m}_{i,n}\hat \epsilon^{1:m}_{i+l,n}|\geq \kappa({0.01}) \hat \sigma_l, 1_0\leq l\leq l_1\big\}$ is empty
	we define $l_n=l_0-1$).

   	\subsection{Covariance estimation}
   	
   	In this section we investigate the finite  sample  properties of the estimators \eqref{Mean-Corrected}  and \eqref{locest} for the covariance matrix $\Sigma_n$
	of a locally stationary process, where we consider
   	\begin{eqnarray}
   	\label{(I)}  \mu(t)&=&2\sin 2\pi(t) ,\\
   	\label{(II)} \mu(t)&=&2-8(t-0.5)^2  ,\\
 	\label{(III)} \mu&=&0 ~,
   	\end{eqnarray}
as  mean functions. Recalling the notation $\FF_i = ( \ldots , \varepsilon_{i-1},\varepsilon_i )   $  we investigate  four different distributions for the errors    in model \eqref{1.1}:
  \begin{description}
	\item (a)  $\{ \epsilon_{i,n}: i=1,\ldots,n\}$ is a stationary     $AR(0.3)$ process with independent  standard normal distributed innovations.
	\item (b)
	$\epsilon_{i,n}=0.8 G(i/n,\FF_i)$ where
$$
 G(t,\FF_i)=0.7\sin(2\pi t) G(t,\FF_i)+\varepsilon_i
 $$
 and $\{\varepsilon_i\}_{i\in \mathbb Z}$ is a sequence of  independent, standardized ($\mathbb E [\varepsilon_i] =0$, Var($\varepsilon_i)=1$) $t$-distributed random variables  with six degrees of freedom.
	\item (c) $\epsilon_{i,n}= G(i/n,\FF_i)$ where
$$
G(t,\FF_i)=\frac{1}{6}(\exp(4(t-0.5)^2)+1)\varepsilon_i+0.6(|\varepsilon_{i-1}|-\E(|\varepsilon_{i-1}|))
$$
 and $\{\varepsilon_i\}_{i\in \mathbb Z}$ is a sequence of  independent standard normal  distributed random variables.
	\item (d) $\epsilon_{i,n}= G(i/n,\FF_i)$ where
$$
G(t,\FF_i)=\frac{1}{4}(\cos(\pi t)+2)(\varepsilon_i+0.9\varepsilon_{i-1}-0.6\varepsilon_{i-2})
$$
	and  $\{\varepsilon_i\}_{i\in \mathbb Z}$  is a sequence of standardized $(\mathbb{E}[\varepsilon_i]=0$, Var$(\varepsilon_i)=1$) 
	independent chi-square distributed random variables with five degrees of freedom.
\end{description}
Note that model (a) defines a stationary process and model (b) defines a locally stationary AR(1) process. Model (c) defines  a nonlinear $tvMA(1)$ process.
Since the innovations $\varepsilon_i$ in model (c) have a symmetric distribution, the covariance matrix of model (c) is diagonal.
Model (d) defines a $tvMA(2)$ process, where  only the entries
in the  diagonal and the first two off diagonals of the covariance matrix do not vanish.

  \begin{table}[htbp]
  	\centering
	\begin{footnotesize}
  	\caption{ \it
  	 Simulated mean squared error $\rho (\hat \Sigma_n - \Sigma_n)$ for the estimators  \eqref{locest} and \eqref{Mean-Corrected} in  model \eqref{1.1}
  	with different mean functions and error processes (a) and (b). }
  	\begin{tabular}{c|c|cc|| cc}
  		&       & \multicolumn{2}{c}{Model (a)} & \multicolumn{2}{c}{Model (b)} \\
      \hline
  	$n$	& $\mu$      &  \eqref{locest} & \eqref{Mean-Corrected} & \eqref{locest} & \eqref{Mean-Corrected}\\
  		\hline
  	 & \eqref{(I)}& 0.952 (0.0104) & 0.637 (0.0105) & 5.034 (0.0311) & 5.532 (0.0083) \\
  		250	& \eqref{(II)} & 0.943 (0.0100)& 0.632 (0.0102) & 5.063 (0.0308) & 5.529 (0.0083) \\
  		& \eqref{(III)} & 0.770 (0.098) & 0.474 (0.0090) & 4.646 (0.0365) & 5.388 (0.0103) \\
  		\hline
  	 &\eqref{(I)} & 0.683 (0.0080) & 0.410 (0.0051) & 4.304 (0.0303) & 5.610 (0.0076) \\
  		500	& \eqref{(II)} & 0.672 (0.0078) & 0.421 (0.0053) & 4.370 (0.0291) & 5.595 (0.0081) \\
  		& \eqref{(III)}  & 0.609 (0.0073) & 0.346(0.0045) & 4.021 (0.0299) & 5.490(0.0096) \\	\hline
  	& \eqref{(I)} & 0.518 (0.0060) & 0.329 (0.0043) & 3.868 (0.0264) & 5.624 (0.0069) \\
  	1000 		&  \eqref{(II)} & 0.535 (0.0062) & 0.322 (0.0043) & 3.881 (0.0265) & 5.632 (0.0070) \\
  		&  \eqref{(III)}  & 0.484 (0.0060) & 0.282 (0.0042) & 3.760 (0.0274) & 5.563 (0.0077) \\
  		\hline
  	\end{tabular}%
  	\label{Estimate1}%
	\end{footnotesize}
  \end{table}%

We examine the estimator for covariance matrix $\Sigma_n$ for sample sizes $n=250$, $500$ and $1000$ using $1000$ simulation runs.
For the estimation of the width  $\l_{n}$ of the band  in \eqref{hneu2} we use  \eqref{criteria-March-2019}  with   $l_0=1$, $l_1=6$.
 In each simulation run the  tuning parameters ($\tau_n$, $b_n$) are determined as described at the beginning of this section. In Table \ref{Estimate1} and \ref{Estimate2}
we display  the simulated mean squared error of the spectral loss $\rho (\hat \Sigma_n - \Sigma_n)$ for different  estimators $\hat \Sigma_n$, where different mean functions  and  error processes in model \eqref{1.1} are considered. In particular we  compare the mean corrected estimator \eqref{locest} for non-stationary error processes with the mean corrected estimator \eqref{Mean-Corrected}
which assumes a stationary error process.
The numbers in brackets show the standard error of the estimates. We observe that in the stationary  model (a) the accuracy of both estimators improve with increasing sample size. Moreover, the estimator \eqref{Mean-Corrected} outperforms \eqref{locest} because this estimator is constructed for stationary processes. On the other hand, for the dependence structures (b) - (d) corresponding to locally stationary processes       the stationary method in \eqref{Mean-Corrected} is not consistent and the estimator \eqref{locest} shows a substantially superior behaviour.

  \begin{table}[htbp]
  	\centering
	\begin{footnotesize}
  	\caption{\it Simulated mean squared error $\rho(\hat \Sigma_n - \Sigma_n)$ for the estimators  \eqref{locest} and \eqref{Mean-Corrected} in  model \eqref{1.1}
  	with different mean functions and error processes (c) and (d) }
  	\begin{tabular}{c|c|cc|| cc}
  		&       & \multicolumn{2}{c}{Model (c)} & \multicolumn{2}{c}{Model (d)} \\
  		\hline
  	$n$	& $\mu$      &  \eqref{locest} & \eqref{Mean-Corrected} & \eqref{locest} & \eqref{Mean-Corrected}\\
  		\hline
  		250 & \eqref{(I)}  & 0.647 (0.0114) & 1.059 (0.0022) & 0.767 (0.0113) & 1.024 (0.0071) \\
  		& \eqref{(II)} & 0.623 (0.0116) & 1.062 (0.0023) & 0.773 (0.011) & 1.037 (0.0071) \\
  		& \eqref{(III)}& 0.557 (0.0109) & 1.045 (0.0023) &0.745 (0.0109) & 1.062 (0.0073) \\
  		\hline
  		500 &\eqref{(I)}& 0.482 (0.0094) & 1.045 (0.0017) & 0.558 (0.010) & 0.963 (0.0045) \\
  		& \eqref{(II)} & 0.478 (0.0094) & 1.043 (0.0016) & 0.569 (0.010) & 0.960 (0.0044) \\
  		& \eqref{(III)} & 0.450 (0.0090) & 1.037 (0.0016) & 0.564 (0.0098) & 0.963 (0.0044) \\
  		\hline
  		1000 & \eqref{(I)}& 0.357 (0.0069) & 1.037 (0.0012) & 0.426 (0.0082) & 0.964 (0.0030) \\
  		& \eqref{(II)} & 0.374 (0.0071) & 1.040 (0.0012) & 0.418 (0.0078) & 0.959 (0.0031) \\
  		& \eqref{(III)} & 0.360 (0.0074) & 1.036(0.0012) & 0.405 (0.0079) & 0.960 (0.0030) \\
  		\hline
  	\end{tabular}
  	\label{Estimate2}%
	\end{footnotesize}
  \end{table}%

  \subsection{Prediction}

  To illustrate the finite sample properties of the estimator proposed in Section \ref{sec4} for prediction
  we examine the mean trend \eqref{(I)}.
 As error process we consider a locally stationary AR(6) model defined by
		
  \begin{equation}
  \label{error_pred}
\prod_{s=1}^6(1-a_s(t)\mathcal B)G(t,\FF_i)=\sigma(t)\varepsilon_i,
 \end{equation}
  where the functions $a_1(t), \ldots , a_6(t)$ are given by
  \begin{align*}
  a_1(t) & =0.6\sin(2\pi( t-0.05)),~
  a_2(t)=0.3\cos^2(3\pi t),~
  a_3(t)=((\exp(t-0.6))^2)/3-0.4,\\
  a_4(t) & =-0.4\sin (6\pi t)-0.1,~a_5(t)=(t-0.3)^2-0.2,~a_6(t)=0.2,
    \end{align*}
 $  \sigma(t)=(1+0.5\sin 2\pi t)^{0.5}$
   and  $\mathcal B$ is the lag operator on the filter $\FF_i$, i.e., $\mathcal B G(t,\FF_i)=G(t,\FF_{i-1}).$ We consider a  standard normal   as well as a $\chi^2(6)$
   distribution for the  errors $\varepsilon_i$ (centered and standardized such that $\mathbb{E}[\varepsilon_i]=0$ Var$(\varepsilon_i)=1$)
  and examine the mean squared error  of the prediction for  sample sizes  $n=250, n=500,n=1000$. 
  We also compare the new predictor with the methods in  \cite{roueff2018prediction}, \cite{kley2019predictive} and \cite{giraud2015aggregation}
  which were theoretically investigated for centered data. In a first step we  used these methods  with  the residuals $\hat \epsilon^{1:m}_{i,n}$ to obtain a prediction for the de-trended series. 
 In a second step  we add to this estimate  the value $\hat \mu^{{1:m} }(m/n)$  to obtain the final prediction of $X_{m+1,n}$.   
  Notice that these authors use time-varying AR$(d)$ processes to approximate the time series for prediction without knowing $d$. Since the error process \eqref{error_pred} is a locally AR$(6)$ process, we investigate the performance of the methods proposed by \cite{roueff2018prediction}, \cite{kley2019predictive} and \cite{giraud2015aggregation}
  for $d=3$, $d=6$ and $d=9$ (note that in the predictor of \cite{kley2019predictive}  $d$ denotes the maximum lag that their algorithm allows). These cases represent  the situation of underestimation, correct-estimation and overestimation of $d$.   Note that   in the cited references there are no rules how   to select $d$.   Moreover, for the method proposed by \cite{kley2019predictive} we choose the parameter $\delta$ in their procedure as $0.05$, as a small parameter $\delta$ prefers the choices of a time-varying model to a stationary model.

\begin{table}[ht]
	\centering
	\begin{footnotesize}
		\caption{\it Simulated mean squared error of different predictors in model \eqref{error_pred} with
		standard normal distributed  $\varepsilon_i$. The numbers in brackets show the standard error and the index
		$*$ represents the predictor with the best best performance.}
	\begin{tabular}{c|cccc||ccc}
		\hline
Method	&\multicolumn{4}{c||}{$t_{pred}=0.5$ }&\multicolumn{3}{c}{$t_{pred}=1$}\\	
	 \hline
	 & lag  & $n=250$ & $n=500$ & $n=1000$ &$n=250$ & $n=500$ & $n=1000$ \\
	\hline
\eqref{hol6}
	& & 1.250& 1.070* &   1.033* & 1.283* &  1.170 & 1.077*\\
	&-& (0.0570)&  (0.0530) & (0.0464) & (0.0596) &   (0.0511) & (0.0464)
	  \\
	\hline
	\multirow{6}{*}{R-S}&{$d=3$} & 1.286& 1.126 & 1.057 & 1.342 & 1.148* & 1.137 \\
&	 & (0.0523) & (0.0499) & (0.0466) & (0.0577) & (0.0589) & (0.0490) \\ 
& 	{$d=6$} & 1.427 & 1.250 & 1.263 & 1.494 & 1.288 & 1.161 \\
&  & (0.0700) & (0.0510) & (0.0532) & (0.0905) & (0.0536) & (0.0518) \\ 
 &	{$d=9$}& 1.895 & 1.297 & 1.209 & 32.286& 1.779 & 1.125 \\
 & & (0.1667) & (0.0566) & (0.0514) & (20.2729) & (0.0630) & (0.0542)\\
\hline
	\multirow{6}{*}{G-R-S}&{$d=3$} & 1.241* & 1.244 & 1.319 & 2.729 & 3.262 & 3.524 \\
&	 & (0.0623) & (0.0607) & (0.0633) & (0.1201) &( 0.1676) &( 0.2425) \\ 
&	{$d=6$} & 1.251 & 1.241 & 1.122 & 2.385 & 2.868 & 2.933 \\
& & (0.0572) & (0.0537) & (0.0552) & (0.1065) & (0.1280) & (0.1378) \\
  &	{$d=9$}& 1.323 & 1.166 & 1.170 & 2.536 & 2.461 & 2.441 \\
&   & (0.0625) &(0.0548) & 0.0500) &(0.1105) & (0.1169)&(0.1311)\\
\hline
	\multirow{6}{*}{K-P-F}&	{$d=3$} & 1.314 & 1.182 & 1.126 & 1.346 & 1.329 & 1.168 \\
	& &(0.0628) & (0.0538) & (0.0484) & (0.0674) &(0.0652) &(0.0517)\\ 
&{$d=6$} & 1.336 & 1.155 & 1.133 & 1.448 & 1.340 & 1.270 \\
&   &(0.0565) & (0.0586) & (0.0474) &(0.0726) &(0.0612) & (0.0503) \\ 
& {$d=9$} & 1.343 & 1.357 & 1.215 & 1.459 & 1.279 & 1.255\\
 &&(0.0598) & (0.0480 )& (0.0509) &(0.0588) &(0.0659) &(0.0581)\\ 
		\hline
	\end{tabular}\label{Table-Normal-Error}
	\end{footnotesize}
\end{table}

\begin{table}[ht]
	\centering
\begin{footnotesize}
	\caption{\it  Simulated mean squared error of different predictors in  model \eqref{error_pred} with (standardized) chi-squared $\varepsilon_i$.
	The numbers in brackets show the standard error and the index
		$*$ represents the predictor with the best best performance.}
	\begin{tabular}{c|cccc||ccc}
		\hline
		&\multicolumn{4}{c||}{$t_{pred}=0.5$ }&\multicolumn{3}{c}{$t_{pred}=1$}\\
		
		\hline
		Method & lag  & $n=250$ & $n=500$ & $n=1000$ &$n=250$ & $n=500$ & $n=1000$ \\
		\hline
		\eqref{hol6}  &-  &  1.201& 1.123 &  1.072*	 &  1.294* &  1.116*& 1.088
		\\            &   &  (0.0577)&  (0.0722)   &   (0.0624)      &  (0.0871)      &      (0.0554)    & (0.0608)
		\\
		\hline
		\multirow{3}{*}{R-S}&{$d=3$}  & 1.276 & 1.032* & 1.100 & 1.307 & 1.196 & 1.061* \\
	&	 & (0.0645) & (0.0757) &(0.0632) & (0.0718) & (0.0794) & (0.0696) \\
	& {$d=6$} & 4.282 & 1.263 & 1.107 & 1.775 & 1.298 & 1.160 \\
&	 & (0.0645) & (0.0787) &(0.0720) &(0.0833) &(0.0868) &(0.0627) \\ 
 &{$d=9$} & 1.726 & 1.347 & 1.159 & 50.111 & 4.181 & 1.210 \\
&  & (0.1022) & (0.0573) &(0.0567) & (24.4556) & (0.0804) & (0.0861) \\ 
\hline
		\multirow{3}{*}{G-R-S}&{$d=3$}  & 1.366 & 1.376 & 1.346 & 2.646 & 3.185 & 3.162 \\
		& & (0.0885) & (0.1016) & (0.0784) & (0.1748) &(0.2806) &(0.3451) \\ 
	&$d=6$  & 1.207 & 1.302 & 1.274 & 2.420 & 2.553 & 2.844 \\
&	 & (0.0651) & (0.0780) & (0.0632) &(0.1217) &(0.2104) & (0.1783) \\ 
		&$d=9$  & 1.263 & 1.299 & 1.182 & 2.338 & 2.722 & 2.721 \\
&		& (0.0618) & (0.0597) & (0.0683) &(0.1440) &(0.2321) & (0.1664) \\ 
\hline
		\multirow{3}{*}{K-P-F}&$d=3$  & 1.120* & 1.101 & 1.176 & 1.372 & 1.320 & 1.061* \\
&		& (0.0668) & (0.0508) & (0.0611) & (0.0731) & (0.0753) & (0.0697) \\ 
	&$d=6$& 1.235 & 1.163 & 1.107 & 1.379 & 1.195 & 1.278 \\
&	 & (0.0644) &(0.0621) & (0.0715) & (0.0946) & (0.0589) & (0.0663) \\ 
		& $d=9$ & 1.134 & 1.283 & 1.202 & 1.317 & 1.293 & 1.132\\
&		 &(0.0712) & (0.0710) &(0.0602) & (0.0793) &(0.0801) &(0.0708)\\ 
		\hline
	\end{tabular}\label{Table-chi-square-Error}
	\end{footnotesize}
\end{table}


 In Table \ref{Table-Normal-Error}   and \ref{Table-chi-square-Error} we present the simulated mean squared error
  $$
  \mathbb{E}[(\hat X^{pred}_{m+1,n} - X_{m+1,n})^2]
  $$
  for the four different prediction methods and different distributions of the innovations.  The columns denoted by $t_{pred}=0.5$ and  $t_{pred}=1$ correspond to a prediction of  $X_{\lfloor n/2 \rfloor +1}$
 from on  $X_{1,1}, \ldots, X_{\lfloor n/2 \rfloor }$ and a prediction of $X_{n,n}$
from  $X_{1,1},\ldots, X_{n-1,n}$, respectively, where 
{we use $l_0= \lceil \log (m) \rceil $ and $l_1=5+\lceil \log (m) \rceil$ in  \eqref{criteria-March-2019}.} 
  The first row shows the simulated mean squared error of the prediction \eqref{hol6}. With increasing sample size this mean squared error approximates $1$.
  This corresponds to  our theoretical  result in Theorem \ref{Thm4}, because we have  for the  model under consideration  $\sigma(0.5)=\sigma(1)=1$.
    The rows denoted by R-S, G-R-S and K-P-F show the simulated mean squared error for predictors proposed by \cite{roueff2018prediction}, \cite{giraud2015aggregation} and \cite{kley2019predictive}, respectively, with different  time lags $d=3,6,9$. In general, the  non-stationary predictor  \eqref{hol6} performs better or similar as the alternative methods with different time lag $d$ in all scenarios.  Our simulation results also demonstrate that the performance of R-S, G-R-S and K-P-F  predictors depend sensitively on the choice  of $d$.   Finally, the large numbers in R-S predictor  is due to the singularity of estimated local covariance matrix. We expect that this can be corrected by using an eigenvalue corrected positive definite covariance matrix estimator similar to  \eqref{Oct4.8}.

  We also  examine the distribution of the prediction error as investigated in Theorem \ref{Thm4}.  For this purpose we show
  in  Figure  \ref{Error-Plot-2}    the  QQ plot of prediction errors of
  the predictors  \eqref{hol6}  for   standard normal distributed errors and  centered and standardized $\chi^2(6)$-distributed
  errors {in} model \eqref{error_pred}, respectively.   The model  is given by \eqref{error_pred} and the sample sizes is $n=1000$.
  These results confirm the theoretical findings in  Theorem  \ref{Thm4}.
\begin{figure}[t]
	\centering
	 \includegraphics[width=8cm,height=6cm]{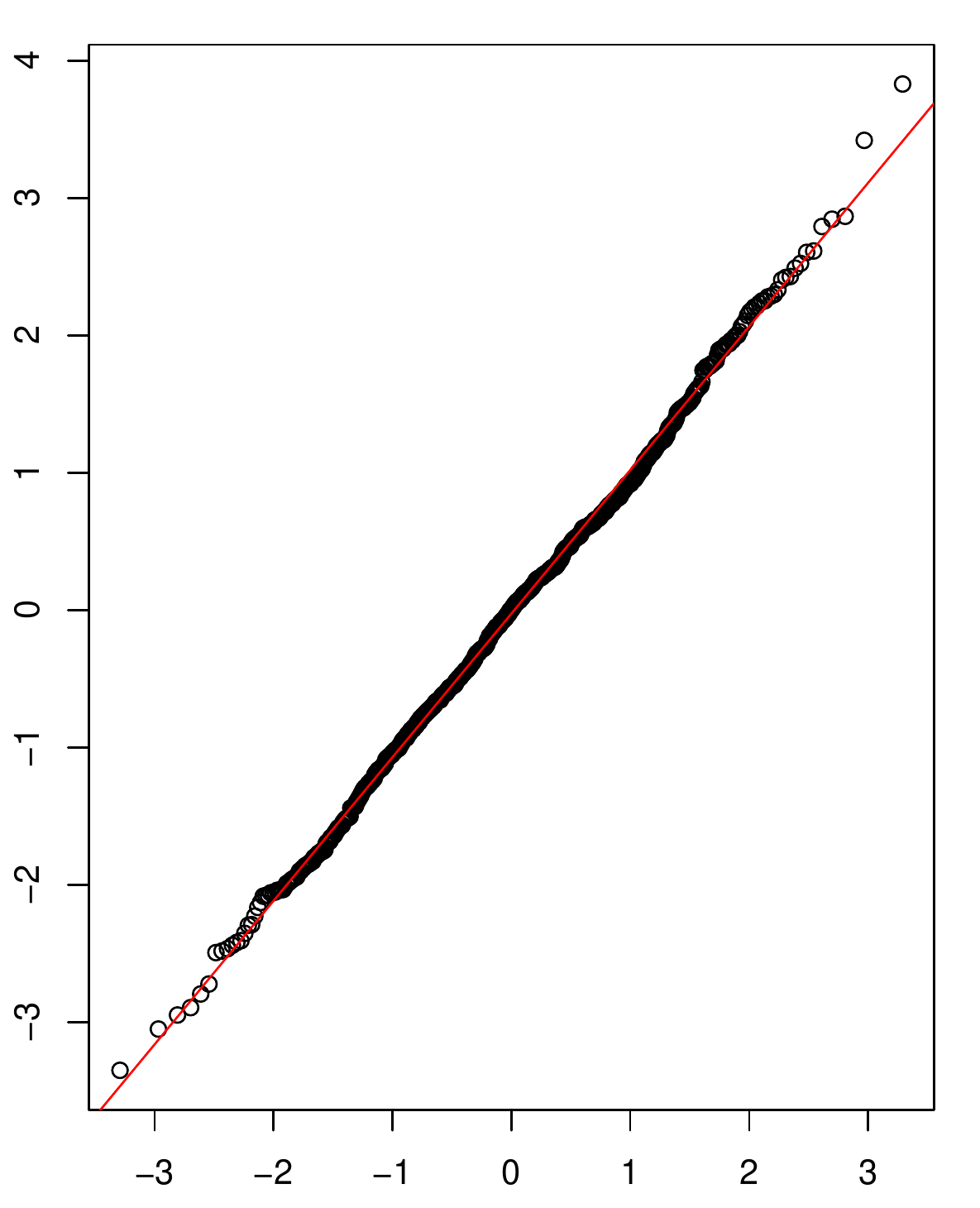}
 \includegraphics[width=8cm,height=6cm]{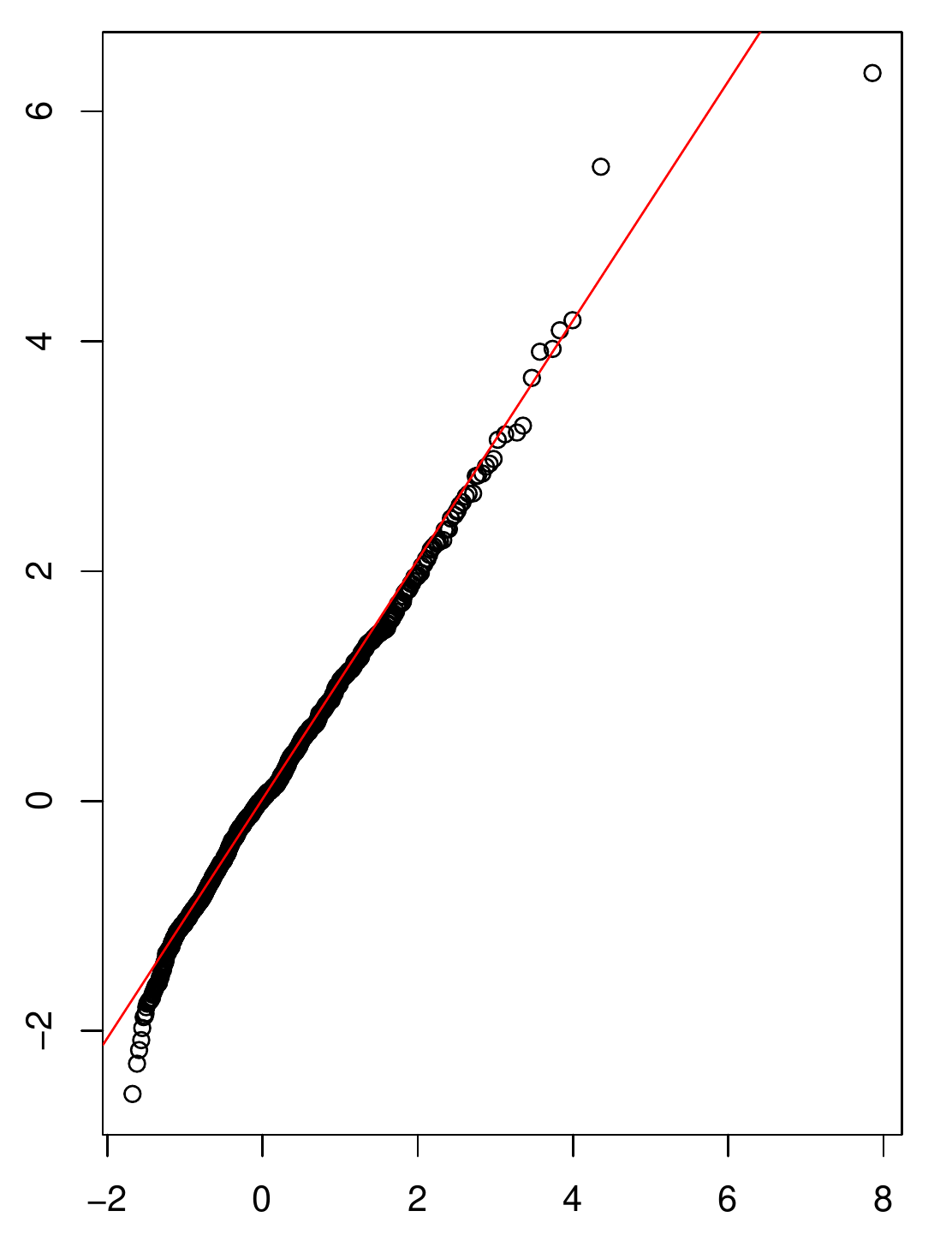}
	\vspace{-.4cm}
	\caption{\it QQ plots  of prediction errors. Left part: standard normal distributed errors. Right part: $(\mathcal{X}^2(6)-6)/\sqrt{12}$-distributed errors.}
	\label{Error-Plot-2}
\end{figure}

Finally, we compare the new predictor \eqref{hol6}
 with the methods proposed by  \cite{roueff2018prediction}, \cite{giraud2015aggregation} and \cite{kley2019predictive} in a  locally stationary MA(6) model
 defined by
 \begin{align}\label{error_pred-MA}
G(t,\FF_i)=\prod_{s=1}^6(1-a_s(t)\mathcal B)\sigma(t)\varepsilon_i,
  \end{align}
  where the time varying coefficients $a_1, \ldots a_{6}$  and the function $\sigma$ are the same as those defined in the locally stationary AR$(6)$ model \eqref{error_pred},
  the mean function is given by  \eqref{(I)} and  the random
  variables $\varepsilon_{i}$ are  independent  standard normal distributed.
The  results  are presented  in  Table  \ref{Table-MA-Error} and we observe similar properties as
in the  locally stationary  AR$(6)$ model \eqref{error_pred}.  A detailed discussion is omitted for the sake of brevity.

\begin{table}[ht]
	\centering
	\begin{footnotesize} 
	\caption{\it  Simulated mean squared error of different predictors with    MA(6) model \eqref{error_pred-MA}. 
	The numbers in brackets show the standard error and the index
		$*$ represents the predictor with the best best performance.}
	\begin{tabular}{c|cccc||ccc}
		\hline
		&\multicolumn{4}{c||}{$t_{pred}=0.5$ }&\multicolumn{3}{c}{$t_{pred}=1$}\\
		\hline
		Method & lag  & $n=250$ & $n=500$ & $n=1000$ &$n=250$ & $n=500$ & $n=1000$ \\
		\hline
	\eqref{hol6} 	&- &  1.187&1.090* &   1.083	 & 1.346* &  1.234* &  1.092*
		\\          &   &  (0.0504)&   (0.0509)   &    (0.0470)        & (0.0627)         &   (0.0554)      &   (0.0511) \\

		\hline
	\multirow{3}{*}{R-S}&$d=3$ & 1.222 & 1.152 & 1.144 & 1.505 & 1.287 & 1.102 \\
	& & (0.0571) & (0.0532) &(0.0503)& (0.0673) & (0.0580) & (0.0475) \\ 
	& $d=6$ & 1.331 & 1.137 & 1.228 & 1.869 & 1.405 & 1.266 \\
	& & (0.0569) & (0.0511)& (0.0519) & (0.0912) & (0.1037) &(0.0511) \\ 
	&$d=9$& 8.757 & 1.338 & 1.138 &  254.780& 2.128 & 1.247 \\
	& & (1.213) &(0.0596) & (0.0515) & (175.380) & (0.1643) &(0.0533) \\ 
	\hline
	\multirow{3}{*}{G-R-S}&$d=3$  & 1.232 & 1.255 & 1.296 & 2.462 & 2.484 & 2.042 \\
    & & (0.0557) & (0.0562) &(0.0642) & (0.1044) &(0.2468) & (0.1060) \\ 
	& $d=6$ & 1.167 & 1.257 & 1.035 & 2.169 & 1.868 & 1.793 \\
	& & (0.0544) &(0.0539) &(0.0492) &(0.0973) & (0.0839) & (0.0849) \\ 
	&$d=9$& 1.128* & 1.178 & 1.087 & 1.985 & 2.064 & 1.949\\
	& & (0.0610) & (0.0604) & (0.0543) &(0.0943) & (0.0925) &(0.0882) \\ 
	\hline
	\multirow{3}{*}{K-P-F}&$d=3$ & 1.286 & 1.280 & 1.051* & 1.571 & 1.404 & 1.292 \\
  &	&(0.0497) & (0.0599) &(0.0456) & (0.0677) & (0.0595) & (0.0545) \\ 
	&$ d=6$ & 1.177 & 1.179 & 1.244 & 1.523 & 1.321 & 1.288 \\
	& & (0.0595) & (0.0538) &(0.0516) & (0.0751) &(0.0669) & (0.0548) \\ 
	&$d=9$& 1.296 & 1.238 & 1.158 & 1.649 & 1.449 & 1.310  \\
&	& (0.0524) & (0.0511) & (0.0479) & (0.0724) &(0.0640) & (0.0606) \\ 
	\hline
		\hline
	\end{tabular}\label{Table-MA-Error}
	\end{footnotesize}
\end{table}

\subsection{Market indices analysis}
In this section we apply our method to predict market indices. Let $p_t$ be the adjusted daily closing value at day $t$, then the log return $r_t$ is defined as
\begin{align*}
r_t=\log p_t-\log p_{t-1}.
\end{align*}
As pointed out by \cite{stuaricua2005nonstationarities},  the sign of $r_t$ is unpredictable. As a result, these authors proposed to model $r_t$ as
\begin{align}\label{Stan}
\log |r_t|=\mu(t)+\sigma(t)\epsilon_t
\end{align}
where $\mu $ and $\sigma $ are time varying functions and $\epsilon_t$ denotes a  zero-mean noise process.  \cite{stuaricua2005nonstationarities} used
 model \eqref{Stan} to study the non-stationarity of stock returns. In this section we apply  the new method  to predict  $y_t:=\log(|r_t|)$ for the  
 SP500, NASDAQ and Dow Jones Index. We consider data from Dec. $19$, $2016$ to Dec. $17$, $2019$. For SP500, NASDAQ and Dow Jones Index, we delete the log return of Jan. 10, 2017, Nov. 13, 2018 and  Nov. 12, 2019 respectively due to their negative infinity values. Therefore the  lengths of the 
 series are $752$.  We 
 use the  new  method to predict the market indices  at
 trading days between April. 8, 2019 and Dec. 17, 2019 for SP500 and NASDAQ and   at 
 trading days between April. 5, 2019 and Dec. 17, 2019 for Dow Jones Series, respectively, and  calculate
 the empirical mean squared error for these predictions. For the sake of comparison we  also apply  the methods of 
 \cite{roueff2018prediction} (R-S),  \cite{giraud2015aggregation}  (G-R-S) and \cite{kley2019predictive} (K-P-F) to the same series. As in the simulation, for fair comparison we perform those algorithms on non-parametrically de-trended data and use the outcome plus $\hat \mu((T-1)/T)$ as the prediction of indices at day $T$.  The corresponding results 
 are listed in Table \ref{DataAna}, where we use the different  lags $3,6,9$ in the procedures based on autoregressive fitting. We observe that the new  prediction 
   method  \eqref{hol6} shows the best performance 
 for all three market indices. For   NASDAQ index the method proposed by  \ \cite{kley2019predictive} with $d=9$ shows a similar performance. In general the parameter $d$ for 
 the prediction  method proposed by  \cite{roueff2018prediction}, \cite{giraud2015aggregation} and \cite{kley2019predictive} is difficult to select, while it has a complicated impact on the predictions when applying those approaches. In Figure \ref{market} we also plot the prediction  error of the different methods  for the three market indices. 
 The left panels display $\log|r_t|$, while the right panels show absolute prediction errors of the prediction 
\eqref{hol6}  and  of  the predictors proposed by \cite{roueff2018prediction} (R-S),  \cite{giraud2015aggregation}  (G-R-S) and \cite{kley2019predictive} (K-P-F) 
 for the corresponding parameter $d \in \{ 3,6,9\} $, which achieves the smallest mean squared error. 

\begin{table}[h]
    \centering
    	\begin{footnotesize}
    \caption{\it Empirical mean squared error of different predictors for SP500, NASDAQ and Dow Jones. The notation $*$ marks  the best method.}
    \begin{tabular}{c|cccc}

        \hline
        Method & lag  &SP500 &NASDAQ & Dow Jones \\
        \hline
        \eqref{hol6}
        &-  & 1.456*& 1.119* &    1.745*
        \\
        \hline
        \multirow{3}{*}{R-S}&d=3 &  1.535 &   1.130 &  1.747 \\
        & d=6 & 1.586 &1.142&  1.873 \\
        &d=9 &  1.607 &1.170 & 1.860 \\
        \hline
        \multirow{3}{*}{G-R-S}&d=3 &  1.817 &   1.826 & 2.054  \\
        &d=6 & 2.689 & 1.350 &  2.361 \\
        &d=9&  2.225 & 1.200 & 2.344\\
        \hline
        \multirow{3}{*}{K-P-F}&d=3 & 1.653 &  1.147 & 1.883\\
        &d=6 & 1.707 &  1.124& 1.938 \\
        & d=9 & 1.763 &1.119* & 1.932\\

        \hline
    \end{tabular}\label{DataAna}
    	\end{footnotesize}
\end{table}

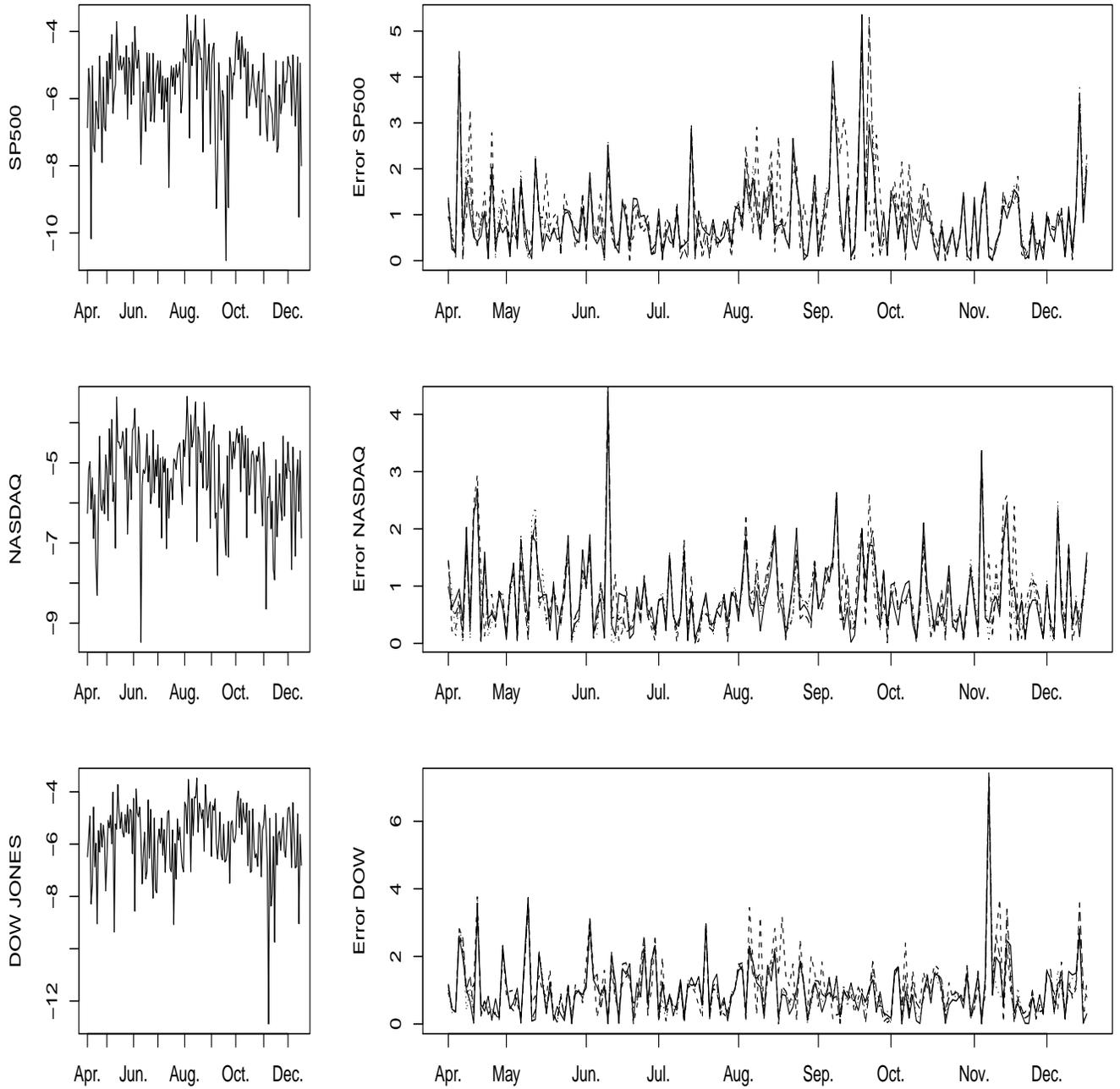
\begin{figure}[H]
        \centering
\tikz[baseline]\draw [solid] (0,0.15) -- (0.6,0.15); (method \eqref{hol6});  \tikz[baseline]\draw [dotted] (0,0.15) -- (0.6,0.15); (R-S);\tikz[baseline]\draw [dashed] (0,0.15) -- (0.6,0.15); 
(G-R-S);\tikz[baseline]\draw[densely dashed] (0,0.15) -- (0.6,0.15); (K-P-F);  
              \vskip -1.cm
 \includegraphics[width=18cm,height=20cm]{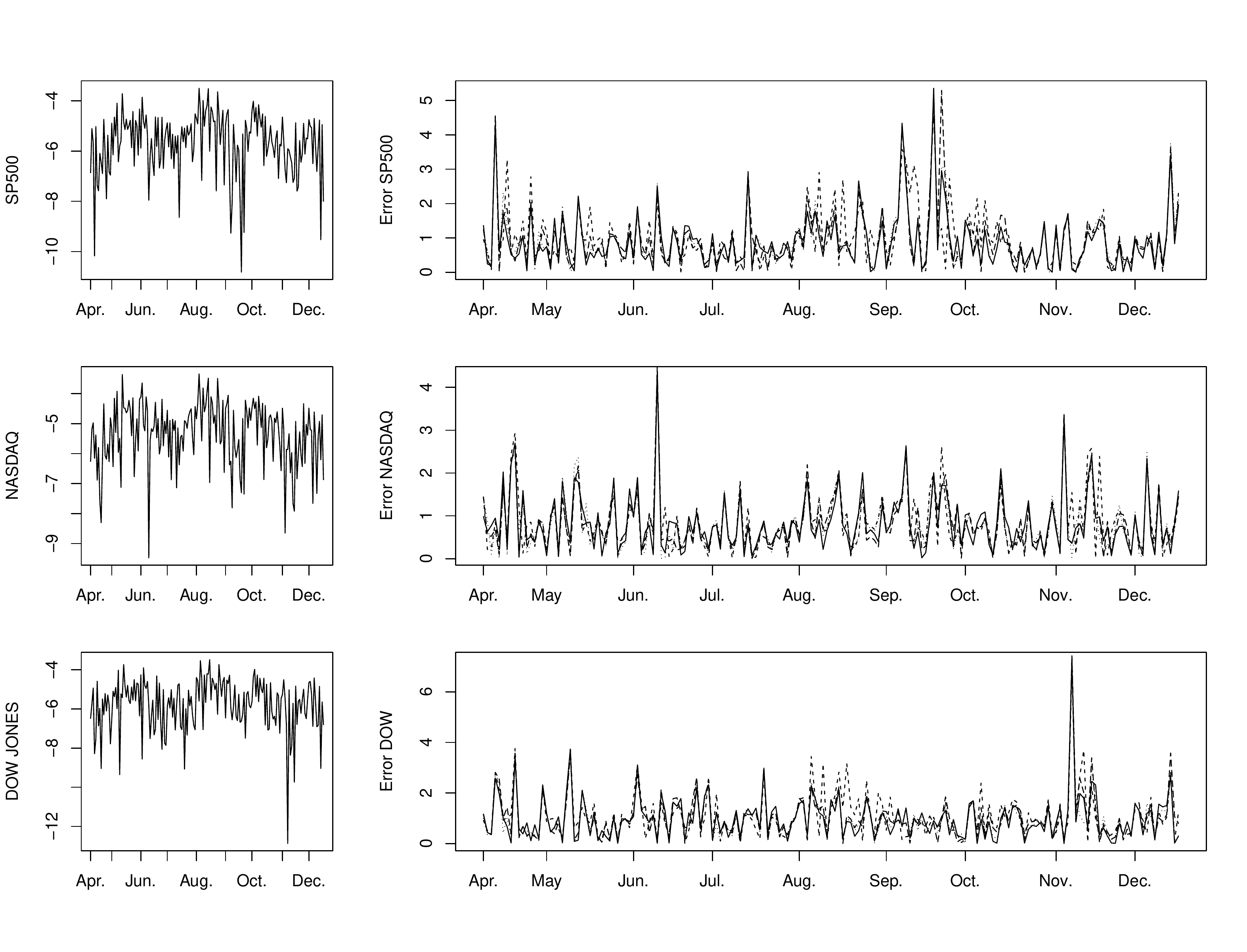} 
      \vskip -1.cm
    \caption{\it  Prediction of   different market indices (left panels).    Right Panel: the absolute prediction errors of the different methods}
  
    \label{market}
\end{figure}

\section{Appendix: Proofs} \label{sec6}

\def\theequation{6.\arabic{equation}}
\setcounter{equation}{0}

 In the proof, we shall use  $\pp_i(\cdot)=\E(\cdot|\FF_i)-\E(\cdot|\FF_{i-1})$ as the projection operator. Let $\epsilon_{i,n}=0$ and $\hat \epsilon_{i,n}=0$ for $i\leq 0$ or $i> n$ for convenience. For a $p-$dimensional real vector  $\mathbf v=( v_1,..., v_p)^\top$, we write $|\mathbf v|=(\sum_{i=1}^pv_i^2)^{1/2}$ for its euclidean norm,
 and write
$\|\mathbf v\|_q= \E\left(| \mathbf v |^q\right)^{1/q}$ if $\mathbf v$ is random.
 Let $M$ denote a sufficiently large constant which varies from line to line.  Write 
 $a\vee b=\max(a,b)$ and $a\wedge b=\min(a,b)$. For positive definite matrix $A$, define $\lambda_{max}(A)$ and 
$\lambda_{min}(A)$ be its largest and smallest eigenvalues, respectively.

\noindent
\subsection{Some auxiliary results}

In this section we provide several auxiliary results, which will be used in the proofs of the main statements.
The main result is Proposition \ref{May-5-prop1}, while  Proposition \ref{Propostiion_4}
  and  \ref{March16_prop5} are   used for a proof of
this statement.

\begin{proposition}\label{Propostiion_4}
If assumptions (L1)-(L3), (M1) hold, $n\tau_n^3\rightarrow \infty$ and $n\tau_n^6=o(1)$, and $\lf cn\rf\leq m\leq n$ for some constant $c, 0<c<1$,
	then the local linear estimate	in \eqref{locallinear} satisfies
		 \begin{align*}
	\sup_{t\in [0,1]}\|\hat \mu^{1:m}(t)-\mu(t)\|_q=O(\tau_n^2+(n\tau_n)^{-1/2}).
	\end{align*}
	\end{proposition}	
	\begin{proof}
		Define the quantities $M_k(t)$, $k=0,1,2$ as 
		\begin{align*}
	M_k(t)=\frac{1}{n\tau_n}\sum_{i=1}^mK\Big (\frac{i/n-t}{\tau_n}\Big )\Big (\frac{i/n-t}{\tau_n}\Big )^k.
		\end{align*}
		The straightforward but tedious calculations by  solving \eqref{localnew} we have for $t\in [0,\frac{m}{n}]$ the solution is
		\begin{align}\label{new.6.2}
		\hat \mu^{1:m}(t)=\frac{1}{n\tau_n}\sum_{i=1}^m \big(\mu \big (\frac{i}{n} \big)+\epsilon_{i,n}\big)	K^*\Big(\frac{i/n-t}{\tau_n}\Big),
		\end{align}
		where 
		\begin{align*}
		K^*\Big(\frac{i/n-t}{\tau_n}\Big)=\frac{M_2(t)K(\frac{i/n-t}{\tau_n})-M_1(t)K(\frac{i/n-t}{\tau_n})(\frac{i/n-t}{\tau_n})}{M_0(t)M_2(t)-M_1^2(t)},
		\end{align*}
		with $0/0=0$ for convenience.
		Observe that $K^*$ is bounded and has a compact support on $[-1,1]$.
	Observing the identity
	\begin{align}\label{new.30}
	\Big\|\sum_{i=1}^m\frac{1}{n\tau_n}K^*
	\Big (\frac{i/n-t}{\tau_n}\Big)\epsilon_{i,n}\Big\|_q=	\Big\|\frac{1}{n\tau_n}\sum_{k=0}^\infty \sum_{i=1}^m\pp_{i-k}K^* \Big (\frac{i/n-t}{\tau_n} \Big)\epsilon_{i,n}\Big\|_q,
	\end{align}
	and applying  Burkholder's  inequality to the martingale difference $\sum_{i=1}^m\pp_{i-k}K^*\big (\frac{i/n-t}{\tau_n}\big )\epsilon_{i,n}$ shows
		\begin{align}\label{new.31}
	\Big\|\sum_{i=1}^m\pp_{i-k}K^* \Big (\frac{i/n-t}{\tau_n} \Big )\epsilon_{i,n}\Big\|^2_q\leq C_0q
	\sum_{i=1}^m\Big\|\pp_{i-k}K^*\Big (\frac{i/n-t}{\tau_n} \Big )\epsilon_{i,n}\Big\|_q^2\leq C_0qn\tau_n\delta_q^2(k)
	\end{align}
	for some   constant $C_0$, where 	
	we have used the same arguments as given in the proof
	 of Theorem 1 in \cite{wu2005nonlinear} for the last inequality, and have used the fact that $m\geq \lf cn\rf$.
	 Combining  \eqref{new.30} and \eqref{new.31} leads to
	 \begin{align}\label{new.32}
	\Big\|\sum_{i=1}^m\frac{1}{n\tau_n}K^* \Big (\frac{i/n-t}{\tau_n} \Big )\epsilon_{i,n}\Big\|_q\leq C^{1/2}_0q^{1/2}(n\tau_n)^{-1/2}\sum_{k=0}^\infty \delta_q(k).
	\end{align}
	Now elementary calculations using condition (M1) with Taylor expansion show that 
		\begin{align}\label{new.67}
	\sup_{t\in[0,1]} \Big |\frac{1}{n\tau_n}\sum_{i=1}^m\mu(\frac{i}{n})	K^*\Big(\frac{i/n-t}{\tau_n}\Big)-\mu(t) \Big |=O(\tau^2_n).
	\end{align}
	Then the the assertion  follows from \eqref{new.6.2}, \eqref{new.32} and \eqref{new.67}.
	\end{proof}

\begin{proposition}\label{March16_prop5}
If assumptions (L1)-(L3), (M1) are satisfied, $n\tau^{3}_n\rightarrow \infty$ and $n\tau_n^6=o(1)$, then
 we have for $1\leq k\leq n$,
	\begin{align*}
	\Big \|\max_{1\leq j\leq n}|\sum_{i=1}^j(\epsilon_{i,n}\epsilon_{i+k,n}-\hat \epsilon_{i,n}\hat \epsilon_{i+k,n})| \Big \|_{q/2}=O(\alpha_{n}),	\end{align*}
	where
	$\alpha_{n}=n\tau_n^3+\tau_n^{-1}+\sqrt{n\tau_n}$.
	\end{proposition}

	\begin{proof}
Proposition 2 follows using similar arguments as given in the  proof of Theorem 3.1 in \cite{dette2015change}.
\end{proof}

\begin{proposition}\label{May-5-prop1} If the assumptions of Theorem \ref{thm2} are satisfied,
and $0\leq k\leq l_n$, there exists a sufficiently large constant $M$ such that
	\begin{align*}
	&\text{(i)\ \ }	\sup_{t\in [0,1]}\|\hat \gamma_k(t)-\gamma_k(t)\|_{q/2}\leq
	M \Big ((nb_n)^{-1/2}+D_kb_n^2+\frac{k}{n}+\frac{\alpha_{n}}{nb_n} \Big),\\
	&\text{(ii)\ \ }	\Big\|\sup_{t\in [0,1]}|\hat \gamma_k(t)-\gamma_k(t)| \Big \|_{q/2}\leq
	M \Big (b_n^{-2/q}(nb_n)^{-1/2}+D_kb_n^2+\frac{k}{n}+\frac{\alpha_{n}}{nb_n} \Big ).
	\end{align*}
\end{proposition}
\begin{proof}
 Without loss of generality, we assume that the lag $k$ is even and define $\tilde \gamma_k(t)$ as the analogue  $\hat \gamma_k(t)$   in \eqref{May5-35},
 where the residuals $\hat \epsilon_{i,n}$ are replaced by the ``true'' errors  $\epsilon_{i,n}$, that is
\begin{align*}
(\tilde \gamma_{k}(t),\tilde \gamma'_{k}(t))^\top=\argmin_{(\beta_0,\beta_1)\in \mathbb R^2}\sum_{i=1}^n\big( \epsilon_{i-k/2,n} \epsilon_{i+k/2,n}-\beta_0-\beta_1(i/n-t)\big)^2K\Big (\frac{i/n-t}{b_n}\Big ).
\end{align*}
Elementary calculations show that
\begin{align}
\label{Solvetildegamma}
\tilde \gamma_k(t)=\frac{\frac{M_2(t)}{nb_n}\sum_{i=1}^n\epsilon_{i-k/2,n}\epsilon_{i+k/2,n}K(\frac{i/n-t}{b_n})-\frac{M_1(t)}{nb_n}\sum_{i=1}^n\epsilon_{i-k/2,n}\epsilon_{i+k/2,n}K(\frac{i/n-t}{b_n})(\frac{i/n-t}{b_n})}{M_0(t)M_2(t)-M^2_1(t)},
\end{align}
where
\begin{align*}
M_k(t)=\frac{1}{nb_n}\sum_{i=1}^nK \Big (\frac{i/n-t}{b_n} \Big) \Big (\frac{i/n-t}{b_n} \Big)^k, \quad k=0,1,2.
\end{align*}
Similarly,  we have
\begin{align*}
\hat \gamma_k(t)=\frac{\frac{M_2(t)}{nb_n}\sum_{i=1}^n\hat \epsilon_{i-k/2,n}\hat \epsilon_{i+k/2,n}K(\frac{i/n-t}{b_n})-\frac{M_1(t)}{nb_n}\sum_{i=1}^n\hat \epsilon_{i-k/2,n}\hat \epsilon_{i+k/2,n}K(\frac{i/n-t}{b_n})(\frac{i/n-t}{b_n})}{M_0(t)M_2(t)-M^2_1(t)}
\end{align*}
and  using the summation by parts formula and Proposition \ref{March16_prop5} it follows that
\begin{align*}
\Big\|\sup_{t\in [0,1]}|\tilde \gamma_k(t)-\hat \gamma_k(t)|\Big\|_{q/2}=O(\frac{\alpha_{n}}{nb_n}).
\end{align*}
 uniformly with respect to  $1\leq k\leq n$
 and it   remains  to show that
	\begin{align*}
&\text{(a)\ \ }	\sup_{t\in [0,1]}\Big\|\tilde \gamma_k(t)-\gamma_k(t)\Big\|_{q/2}\leq M \Big((nb_n)^{-1/2}+D_kb_n^2+\frac{k}{n}\Big),\\
&\text{(b)\ \ }\Big	\|\sup_{t\in [0,1]}|\tilde \gamma_k(t)-\gamma_k(t)|\Big\|_{q/2}\leq  M \Big( b_n^{-2/q}(nb_n)^{-1/2}+D_kb_n^2+\frac{k}{n}\Big).
\end{align*}
Let $\eta_{i,k}=\epsilon_{i-k/2,n}\epsilon_{i+k/2,n}$ (note that $k$ is even).  By \eqref{Solvetildegamma} we have
\begin{align*}
    \tilde \gamma(t)=\tilde M_1(t)\big(\frac{1}{nb_n}\sum_{i=1}^n\eta_{i,k}K\big(\frac{i/n-t}{b_n}\big)\big)+\tilde M_2(t)\big(\frac{1}{nb_n}\sum_{i=1}^n\eta_{i,k}K\big(\frac{i/n-t}{b_n}\big)\big(\frac{i/n-t}{b_n}\big)\big)
\end{align*}
with $\tilde M_1(t)=\frac{M_2(t)}{M_0(t)M_2(t)-M^2_1(t)}$, $\tilde M_2(t)=\frac{-M_1(t)}{M_0(t)M_2(t)-M^2_1(t)}$. Notice that
\begin{align*}
    \gamma(t)=\tilde M_1(t)\big(\frac{1}{nb_n}\sum_{i=1}^n\gamma_k(t)K\big(\frac{i/n-t}{b_n}\big)\big)+\tilde M_2(t)\big(\frac{1}{nb_n}\sum_{i=1}^n\gamma_k(t)K\big(\frac{i/n-t}{b_n}\big)\big(\frac{i/n-t}{b_n}\big)\big)
\end{align*}
As a result, we can decompose $\tilde \gamma_k(t)-\gamma_k(t)$ into a random part and a deterministic part, i.e.
\begin{align*}
\tilde \gamma_k(t)-\gamma_k(t)=\Xi^d_k(t)+\Xi^s_k(t),
\end{align*}
where
\begin{align*}
 \Xi^d_k(t)=\tilde M_1(t) \frac{1}{nb_n}&\sum_{i=1}^n(\E \eta_{i,k}-\gamma_k(t))K\big(\frac{i/n-t}{b_n}\big) \notag\\&+\tilde M_2(t) \frac{1}{nb_n}\sum_{i=1}^n(\E \eta_{i,k}-\gamma_k(t))K\big(\frac{i/n-t}{b_n}\big)\big(\frac{i/n-t}{b_n}\big) ,\\
   \Xi^s_k(t)=\tilde M_1(t) \frac{1}{nb_n}&\sum_{i=1}^n(\eta_{i,k}-\E \eta_{i,k})K\big(\frac{i/n-t}{b_n}\big) \notag\\&+\tilde M_2(t) \frac{1}{nb_n}\sum_{i=1}^n(\eta_{i,k}-\E \eta_{i,k})K\big(\frac{i/n-t}{b_n}\big)\big(\frac{i/n-t}{b_n}\big) .
\end{align*}
To complete the proof we will show that  (uniformly for $0\leq k\leq l_n$)
\begin{align}
    &	\sup_{t\in [0,1]}|\Xi^d_{k}(t)|\leq M \left( D_kb_n^2+ {\frac{k}{n}+\frac{1}{nb_n}}\right),\label{Inequal1}\\
    &	\sup_{t\in [0,1]}\|\Xi^s_{k}(t)\|_{q/2}\leq M \left( (nb_n)^{-1/2}\right),
    \label{Inequal2}\\
    &\|\sup_{t\in [0,1]}|\Xi^s_{k}(t)|\|_{q/2}\leq M \left( b_n^{-2/q}(nb_n)^{-1/2}\right).
    \label{Inequal3}
\end{align}
Observe that $\Xi_{k}^d(t)$  can be further decomposed as
\begin{align*} 
\Xi^d_{k}(t)=\Xi^{d}_{1,k}(t)+\Xi^{d}_{2,k}(t),
\end{align*}
where
\begin{align*}
 \Xi^d_{1,k}(t)=\tilde M_1(t) \frac{1}{nb_n}&\sum_{i=1}^n(\E \eta_{i,k}-\gamma_k(i/n))K\big(\frac{i/n-t}{b_n}\big) \notag\\&+\tilde M_2(t) \frac{1}{nb_n}\sum_{i=1}^n(\E \eta_{i,k}-\gamma_k(i/n))K\big(\frac{i/n-t}{b_n}\big)\big(\frac{i/n-t}{b_n}\big) ,\\
   \Xi^d_{2,k}(t)=\tilde M_1(t) \frac{1}{nb_n}&\sum_{i=1}^n(\gamma_k(i/n)-\gamma_k(t))K\big(\frac{i/n-t}{b_n}\big) \notag\\&+\tilde M_2(t) \frac{1}{nb_n}\sum_{i=1}^n(\gamma_k(i/n)-\gamma_k(t))K\big(\frac{i/n-t}{b_n}\big)\big(\frac{i/n-t}{b_n}\big) .
\end{align*}
By conditions (L1), (L2) and  a Taylor expansion it follows that
\begin{align*}
|\E(\eta_{i,k})-\gamma_k(i/n)|\notag=
&\Big|\E\Big(G\Big(\frac{i-k/2}{n},\FF_0\Big)G\Big(\frac{i+k/2}{n},\FF_k\Big)\Big)-\E\Big(G\Big(\frac{i}{n},\FF_0\Big)G\Big(\frac{i}{n},\FF_k\Big)\Big)\Big|\notag\\
=&O(k/n)
\end{align*}
uniformly with respect to  $i$.
A straightforward but tedious calculation now shows that
\begin{align}\label{Xid_1}
\sup_{t\in[0,1]}|\Xi^{d}_{1,k}(t)|= {O(k/n+{\frac{1}{nb_n}})}
\end{align}as $n\rightarrow \infty$,
uniformly with respect to $0\leq k\leq l_n$.
In addition  by condition (A1), we obtain that
\begin{align}\label{Xid_2}
\sup_{t\in[0,1]}|\Xi^{d}_{2,k}(t)|=O(D_kb_n^2 {+\frac{1}{nb_n}})
\end{align}
(uniformly for $0\leq k\leq l_n$).
As a result, inequality \eqref{Inequal1} follows from \eqref{Xid_1} and \eqref{Xid_2}.
For $\Xi^s_k(t)$, an application of
the   Cauchy-Schwartz inequality shows  that
\begin{align*} 
    \|\mathcal P_{i+k/2-s}\eta_{i,k}\|_{q/2}\leq M(\delta_q(s)+\mathbf 1(s\geq k)\delta_q(s-k)),
\end{align*}
(uniformly with respect to  $i$)
and   assertion  \eqref{Inequal2} now follows using  similar  arguments as given in the  proof of Proposition \ref{Propostiion_4}.
 By Assumption (K)
 and  similar  arguments as given in the  proof of Proposition \ref{Propostiion_4} we have
\begin{align}
   \sup_{t\in [0,1]}\Big\|\frac{\partial}{\partial t}\Xi^s_{k}(t)\Big\|_{q/2}\leq M\left( (nb_n)^{-1/2}b_n^{-1}\right)
	\label{Inequal4}
\end{align}
(uniformly with respect  $0\leq k\leq l_n$).
Finally, inequality \eqref{Inequal3}    follows from \eqref{Inequal2}, \eqref{Inequal4}  and Proposition B.1  in  \cite{dette2015change}, which completes the proof.
\end{proof}

\subsection{Proof of Theorem \ref{thm1} and \ref{thm2}}

For the sake of brevity we restrict ourselves to the proof of Theorem \ref{thm2}.
Theorem \ref{thm1}  can be shown by similar but substantially simpler arguments.

Define the banded matrix  $\Sigma_{l_{n},n}:=(\sigma_{i,j}\mathbf 1(|i-j|\leq l_n))$, where we use  the  symbol $\sigma_{i,j}$  for  $\sigma_{i,j,n}$
to simplify the notation. Note that $\Sigma_{l_{n},n}-\Sigma_n$ is a symmetric matrix and by   Gershgorin's circle theorem
it follows that
\begin{align}\label{April13-37}
\rho(\Sigma_{l_{n},n}-\Sigma_n) & \leq \max_{1\leq i\leq n} \sum_{j=1}^n|\sigma_{i,j}-\sigma_{i,j}\mathbf 1(|i-j|\leq l_n)|
\notag\\
& \leq \max_{1\leq i\leq n}\Big(\sum_{j=1}^{( i-l_n )\vee 1}|\sigma_{i,j}|+\sum_{j=( i+l_n)\wedge n}^n|\sigma_{i,j}| \Big).
\end{align}
Using similar arguments as given in the proof  of
Lemma 5 of \cite{zhou2010simultaneous}  it follows that
\begin{align}\label{small-cov}
|\sigma_{i,j}|=O\Big(\sum_{s=1}^\infty\chi^{2s+|i-j|}\Big)=O(\chi^{|i-j|})~,
\end{align}
for all $i,j\in \mathbb N$,  and straightforward calculations give
\begin{align*}
\max_{1\leq i\leq n} \Big ( \sum_{j=1}^{( i-l_n )\vee 1}|\sigma_{i,j}|
\Big  )=O(\chi^{l_n})~, ~~\max_{1\leq i\leq n} \Big  (\sum_{j=( i+l_n)\wedge n}^n|\sigma_{i,j}|
\Big )=O(\chi^{l_n}).
\end{align*}
Therefore we obtain from \eqref{April13-37} the estimate
\begin{align*} 
\rho(\Sigma_{l_n,n}-\Sigma_n)=O(\chi^{l_n}).
\end{align*}
 Note that, by definition, $\sigma_{i,j}=\E(G(\frac{i}{n},\FF_i)G(\frac{j}{n},\FF_j))$, $\gamma_{|i-j|} (\frac{i+j}{2n}) =\E(G(\frac{i+j}{2n},\FF_i)G(\frac{i+j}{2n},\FF_j))$,
 then using conditions (L1), (L2) we have
\begin{align} \label{new.5.16}
\max_{|i-j|\leq l_n}\Big|\gamma_{|i-j|}(\frac{i+j}{2n})-\sigma_{i,j}\Big|\leq M\frac{l_n}{n}
\end{align}
for some large constant $M$.

On the other hand, similarly to \eqref{April13-37} it follows that
\begin{align}
\rho(\hat \Sigma_{n}-\Sigma_{l_n,n}) &\leq \max_{1\leq i\leq n} \sum_{j=1}^n|\big(\sigma_{i,j}-\hat \gamma_{|i-j|}(\frac{i+j}{2n})
\big)\mathbf 1(|i-j|\leq l_n)|
\notag\\
& =\max_{1\leq i\leq n}
\Big (\sum_{j=( i-l_n )\vee 1}^{j=( i+l_n)\wedge n}|\hat \gamma_{|i-j|}(\frac{i+j}{2n})-\sigma_{i,j}|
\Big ). \label{hg1}
\end{align}
By Proposition  \ref{May-5-prop1} it follows that
\begin{align}\label{New.5.20}
   & \Big \|\max_{1\leq i\leq n}
   \sum_{j=( i-l_n )\vee 1}^{j=( i+l_n)\wedge n}|\hat \gamma_{|i-j|}(\frac{i+j}{2n})-\sigma_{i,j}|\Big \|_{q/2}\notag\\\leq& \Big \|\sum_{j=( i-l_n )\vee 1}^{j=( i+l_n)\wedge n}\max_{1\leq i\leq n}|\hat \gamma_{|i-j|}(\frac{i+j}{2n})-\gamma_{|i-j|}(\frac{i+j}{2n})| \Big \|_{q/2}+\max_{1\leq i\leq n}\sum_{j=( i-l_n )\vee 1}^{j=( i+l_n)\wedge n}| \gamma_{|i-j|}(\frac{i+j}{2n})-\sigma_{i,j}|\notag \\
   \leq &
   { M\Big(l_n(b_n^{-2/q}(nb_n)^{-1/2}+\frac{\alpha_n}{nb_n})+\frac{l_n^2}{n}+\sum_{i=0}^{l_n}D_ib_n^2\Big),}
\end{align}
where the quantities $D_k$ are defined in (A1), for which we have used \eqref{new.5.16}
and the estimate
\begin{align*} 
    \Big\|\sum_{j=( i-l_n )\vee 1}^{j=( i+l_n)\wedge n}\max_{1\leq i\leq n}|&\hat \gamma_{|i-j|}(\frac{i+j}{2n})-\gamma_{|i-j|}(\frac{i+j}{2n})|\Big\|_{q/2}\notag\\
  & = O\Big(\sum_{k=0}^{l_n}\Big(b_n^{-2/q}(nb_n)^{-1/2}+D_kb_n^2+\frac{k}{n}+\frac{\alpha_{n}}{nb_n}\Big)\Big)
  \notag\\& =O\Big(l_n\Big(b_n^{-2/q}(nb_n)^{-1/2}+\frac{l_n}{n}+\frac{\alpha_{n}}{nb_n}\Big)+\sum_{k=0}^{l_n} D_kb_n^2\Big)
\end{align*}
Therefore the theorem follows from  \eqref{hg1} and \eqref{New.5.20}. 

\subsection{Proof of Corollary \ref{Corol2}}
Condition (E1) shows that the quantity
\begin{align*}
W=\Sigma_{n,m}^{-1/2}\hat \Sigma_{n,m}\Sigma^{-1/2}_{n,m}.
\end{align*}
 is well defined.  By our construction, $W$  is positive definite {with probability tending to $1$}.
Then by \eqref{rnm}  and condition (E1), we have that
\begin{align*}
\|\rho(W-I_{m\times m})\|_{q/2}=O(r_n),
\end{align*}
where $I_{m\times m}$ is an $m\times m$ diagonal matrix. Now the corollary follows from the argument in the proof of Theorem 2 of  \cite{mcmurry2010banded} and the 
fact that $\rho(\Sigma_{n,m})$ is bounded which is a consequence of Gershgorin's circle theorem.


\subsection{Proof of Theorem \ref{Thm4}}
By the projection theorem,   equation \eqref{April 6-33} is equivalent to
\begin{align*}
\E(X_{m+1,n}- X^{\rm pred}_{m+1,n})=0, \quad 
\E((X_{m+1,n}-  X^{\rm pred}_{m+1,n})X_{j,n})=0 ; \quad j=1,\ldots,m.  
\end{align*}
Using these equations in  \eqref{April-6-hatX} yields
\begin{align}\label{Aug-25-2019}
& a_{m+1,n}
=\mu\left(\frac{m+1}{n}\right)-\sum_{s=1}^ma_{m+1-s,n}\mu\left(\frac{s}{n}\right) ~, \\
& \E\Big [ \Big ( G\Big(\frac{m+1}{n},\FF_{m+1}\Big)-\sum_{s=1}^ma_{m+1-s,n}G\Big(\frac{s}{n},\FF_s\Big)\Big )G\Big(\frac{j}{n},\FF_j\Big)\Big ] =0  \nonumber
\end{align}
($1\leq j\leq m$),   {which shows that the vector ${\mathbf a}_m^*$ in \eqref{April 6-33} is given by
\begin{align}\label{star_a}
\mathbf a_m^*=\Sigma_{n,m}^{-1}      \boldsymbol{ \gamma}_m,
\end{align}
where $  \boldsymbol \gamma_m = ( \sigma_{m+1,1} , \ldots , \sigma_{m+1,m} )^{\top }$.
Let
 $$
\boldsymbol \gamma_{m,l_n} = ( 0, \ldots ,  0 , \sigma_{m+1,m-l_{n}+1} , \ldots  ,  \sigma_{m+1,m}  )^{\top }
$$
be the vector with  $j_{th}$ entry  given by  $\sigma_{m+1,j}\mathbf 1(m+1-j\leq  l_n)$. By the representation of $\mathbf{\hat a}_m^*$ in \eqref{hat_a_star}, we have
\begin{align*}
\hat {\mathbf a}_m^*-\mathbf a_m^*=G_{1} + G_{2} + G_{3} ,
\end{align*}
where the terms $G_{1}$, $ G_{2}$ and $G_{3}$ are defined by
\begin{align*}
G_{1} & =\hat \Sigma_{n,m}^{-1}(\bs {\hat \gamma}_{n}^{1:m}-\bs\gamma_{m,l_n}), \\
G_{2}  &=(\hat \Sigma_{n,m}^{-1}-\Sigma_{n,m}^{-1})\bs \gamma_{m,l_n}, \\
G_{3} &=\Sigma_{n,m}^{-1}(\bs \gamma_{m,l_n}-\bs\gamma_m).
\end{align*}
In the following we shall show that $G_{j}  = O_{\mathbb{P}}(r_n)$ for $j=1,2,3$, which implies
 \begin{align}\label{bound-astar}
|\hat {\mathbf a}_m^*-\mathbf a_m^*|=O_{\mathbb{P}}(r_n).
\end{align}
Using similar arguments as given in  the derivation of \eqref{New.5.20}
we have
 \begin{align*}
\|\hat {\bs \gamma}^{1:m}_{n}-\bs \gamma_{m,l_n}\|_{q/2}=\Big\|\Big(\sum_{s= m+1- l_n}^{m}|\hat \gamma_{m+1-s}^{1:m}\big(\frac{m+s}{2n}\big)
-\sigma_{m+1,s}|^2\Big)^{1/2}\Big\|_{q/2}=O(r_n).
\end{align*}
A straightforward calculation using assumption (E1) and Corollary \ref{Corol2} show 
$$
G_{1}
 \leq |\rho(\hat \Sigma_{n,m}^{-1})|| \bs{\hat \gamma}_{n}^{1:m}-\bs{\gamma}_{m,l_n} | = O_{\mathbb{P}}(r_n).
$$ 
{By \eqref{small-cov} $|\bs{\gamma}_{m,l_n}|$ is bounded. By Corollary \ref{Corol2}  it also follows $G_{2}=O_{\mathbb{P}}(r_n)$.} Observing  \eqref{small-cov} we obtain
\begin{align}\label{gammamln}
|\bs{\gamma}_{m,l_n}-\bs{\gamma}_m|=\big(\sum_{j=1}^{ m- l_n}\sigma^2_{m,j}\big)^{1/2}\leq M\chi^{l_n}
\end{align}
which implies $G_{3}=O(r_n)$, and hence  \eqref{bound-astar} follows.
For a proof of part (a), it now  remains to show that
\begin{align}\label{boundamplus1}
|\hat a_{m+1,n}-a_{m+1,n}|=O_{\mathbb{P}}(r_n^\circ).
\end{align}
From  \eqref{Aug-25-2019} and definition \eqref{hatam} it follows that 
\begin{align}
\hat a_{m+1,n}-a_{m+1,n}&=\hat \mu^{1:m}\big(\frac{m}{n}\big)-\mu\big(\frac{m+1}{n}\big)+\Big(\sum_{s=1}^ma_{m+1-s,n}\mu\big(\frac{s}{n}\big)-\sum_{s=1}^m\hat a_{m+1-s,n}\hat \mu^{1:m}\big(\frac{s}{n}\big)\Big)\notag\\
&=\Big(\hat \mu^{1:m}(\frac{m}{n})-\mu(\frac{m+1}{n})\Big)+\sum_{s=1}^ma_{m+1-s,n}\Big(\mu(\frac{s}{n})-\hat \mu^{1:m}(\frac{s}{n})\Big)\notag\\&+\sum_{s=1}^m\hat \mu^{1:m}\big(\frac{s}{n}\big)\big(a_{m+1-s,n}-\hat a_{m+1-s,n}\big)\notag\\&:=
H_{1} + H_{2} + H_{3},\notag
\end{align}
where the statistics $H_{1}$, $H_{2} $ and $H_{3}$ are defined in an obvious way.
Using assumption (M1) and Proposition \ref{Propostiion_4}, we have that $\|H_{1}\|_{q}=O(\tau_n^2+(n\tau_n)^{-1/2})$.
For an estimate of  $H_2$ we need to determine the order of $\mathbf a_m^*$ defined in \eqref{April 6-33}.
 For this purpose we  define
\begin{align}\notag
\Sigma_{n,m,l_n}=(\sigma_{i,j,n}\mathbf 1(|i-j|\leq l_n))_{1\leq i,j\leq n},\quad \mathbf a_{m,l_n}^*=\Sigma_{n,m,l_n}^{-1}\bs \gamma_{m,l_n}~, 
\end{align}
then using \eqref{small-cov} and \eqref{gammamln} we get
\begin{align}
|\mathbf a_{m,l_n}^*-\mathbf a_m^*|=O(\chi^{l_n}).\label{amlnstar}
\end{align}
Denote by $a_{m,l_n,j}$, $\gamma_{m,l_n,j}$ the $j_{th}$ entry of  the vector $\mathbf a_{m,l_n}^*$ and $\bs \gamma_{m,l_n}$, respectively. 
Define \begin{align} \notag
H_{2,l_n}=\sum_{s=1}^ma_{m,l_n,m+1-s}\Big(\mu(\frac{s}{n})-\hat \mu^{1:m}(\frac{s}{n})\Big),
\end{align}
then, by \eqref{amlnstar} and Proposition \ref{Propostiion_4}, it follows  that
\begin{align}\label{H2ln1}
H_{2,l_n}-H_2=O_{\mathbb{P}}(\sqrt n\chi^{l_n}(\tau_n^2+(n\tau_n)^{-1/2})).
\end{align}
Hence it suffices to study the order of $H_{2,l_n}$. 
Denote the $(i,j)_{th}$ entry of  the matrix $\Sigma_{n,m,l_n}^{-1}$ by $\Sigma_{n,m,l_n}^{-1}(i,j)$. Since $\Sigma_{n,m,l_n}$ is $l_n$-banded, $\lim_{n\rightarrow \infty}\|\Sigma_{n,m,l_n}\|_F<\infty$ and condition $(E1)$, we can apply Proposition 2.2 of \cite{demko1984decay}, and obtain 
\begin{align}\label{inversesmall}
|\Sigma_{n,m,l_n}^{-1}(i,j)|\leq C_nq_n^{\frac{2|i-j|}{l_n}}~,
\end{align}
where $q_n=(\sqrt {r_n}-1)/(\sqrt {r_n}+1)$, $r_n=\lambda_{max}(\Sigma_{n,m,l_n})/\lambda_{min}(\Sigma_{n,m,l_n})$, $C_n
=\max(\lambda^{-1}_{min},C_{0n})$, $C_{0n}=(1+r_n^{1/2})^2/(2\lambda_{min}(A)r_n)$.  By condition $(E1)$ and \eqref{small-cov},  it follows that there 
exists a positive constant $M$ and a constant $Q \in ( 0,1)$ such that 
$$C_n\leq M, \quad0<q_n\leq Q<1.$$ 
Then, if $h$ a is positive constant such that $nQ^{h\log n}l_n^{1/2}=O(\log^{1/2} n)$, we have uniformly for $1\leq i\leq m-hl_n\log n$ 
\begin{align*}
a_{m,l_n,i} & =\sum_{j=1}^m\Sigma_{n,m,l_n}^{-1}(i,j)\gamma_{m,l_n,j}=\sum_{j=m-l_n+1}^m\Sigma_{n,m,l_n}^{-1}(i,j)\gamma_{m,l_n,j}\\
& =O(l_nQ^{h\log n})=O\Big (\frac{l^{1/2}_n\log^{1/2} n}{n} \Big ).
\end{align*}
On the other hand, observing the fact $|\mathbf a_m^*|\leq |\rho(\Sigma_{n,m}^{-1})| |\bs{ \gamma}_m|<\infty$ yields 
\begin{align}
\label{boundastarnorm}\sum_{i\in(m-hl_n\log n,m]}a^2_{m,l_n,i}\leq M'<\infty
\end{align}
for some constant $M'$. Thus it follows from  Proposition \ref{Propostiion_4} and an application of  the Cauchy Schwarz inequality that
\begin{align}\label{H2ln2}
|H_{2,l_n}|\leq & \Big 
|\sum_{s=1}^ma_{m,l_n,m+1-s}\Big(\mu(\frac{s}{n})-\hat \mu^{1:m}(\frac{s}{n})\Big)\mathbf 1(s\geq hl_n\log n+1)\Big |\notag\\
&+ \Big  |\sum_{s=1}^ma_{m,l_n,m+1-s}\Big(\mu(\frac{s}{n})-\hat \mu^{1:m}(\frac{s}{n})\Big)\mathbf 1(s< hl_n\log n+1)\Big  |\notag\\
&=O_{\mathbb{P}}(l^{1/2}_n\log^{1/2} n(\tau_n^2+(n\tau_n)^{-1/2})).
\end{align}
Equation \eqref{H2ln1} and \eqref{H2ln2} now show that $H_{2}=O_{\mathbb{P}}(r_n^\circ)$, where $r_n^\circ$ is defined in \eqref{rno}. 
Finally, for the estimate of   $H_3$ we define 
$$H_{3,l_n}=\sum_{s=1}^m\hat \mu^{1:m}\big(\frac{s}{n}\big)\big(a_{m,l_n,m+1-s}-\hat a_{m+1-s,n}\big).$$ By \eqref{amlnstar} we find $|H_{3,l_n}-H_{3}|=O_{\mathbb{P}}(\sqrt n\chi^{l_n})$. Notice that \eqref{bound-astar} and \eqref{amlnstar} yield that $|\hat {\mathbf a}_m^*-\mathbf a_{m,l_n}^*|=O_{\mathbb{P}}(r_n)$. Furthermore, similarly to \eqref{inversesmall},  using 
 Proposition 2.2 of \cite{demko1984decay}      it follows that there exist  constants $M_0>0$ and $Q_0 \in (0,1)$ such that
$$
|\hat \Sigma_{n,m}^{-1}(i,j)|\leq M_{0} Q_{0}^{\frac{2|i-j|}{l_n}} ~,~~ 1\leq i,j\leq m
$$
 with probability tending to $1$.
Using this fact and similar arguments as for the derivation of 
 \eqref{H2ln2}, we obtain $|H_{3,l_n}|=O_{\mathbb{P}}((l_n^{1/2}\log^{1/2} n)r_n)=O_{\mathbb{P}}(r_n^\circ)$.
This proves \eqref{boundamplus1}  and completes the proof of part  (a).

\noindent
For a proof of part  (b), we recall the  definition of the filter $G$ in \eqref{LocAR}  and obtain
\begin{align*}
	G(\tfrac{m+1}{n},\FF_{m+1})=\sum_{s=1}^pa_s(\tfrac{m+1}{n})G(\tfrac{m+1}{n},\FF_{m+1-s})+\sum_{s=p+1}^{m}d_sG(\tfrac{s}{n},\FF_{m+1-s})+ {\sigma(\tfrac{m+1}{n})}\varepsilon_{m+1},
	\end{align*}
	where $d_s=0$, for $ p+1\leq s\leq m$.
	Observe that
	\begin{eqnarray}  \nonumber
	\E\big(G(\tfrac{m+1}{n},\FF_{m+1})G(\tfrac{j}{n},\FF_j)\big)&=&\sum_{s=1}^pa_s(\tfrac{m+1}{n})\E\big(G(\tfrac{m+1}{n},\FF_{m+1-s})G(\tfrac{j}{n},\FF_j)\big)\\
&+&\label{New-Cov-1}
\sum_{s=p+1}^{m}d_s\E\big(G(\tfrac{s}{n},\FF_{m+1-s})G(\tfrac{j}{n},\FF_j)\big),
\end{eqnarray}
($1\leq j\leq m-p$), and 
	\begin{eqnarray}
 \nonumber \E(G\big(\tfrac{m+1}{n},\FF_{m+1})G(\tfrac{m+1}{n},\FF_j)\big)&=&\sum_{s=1}^pa_s(\tfrac{m+1}{n})\E\big(G(\tfrac{m+1}{n},\FF_{m+1-s})G(\tfrac{m+1}{n},\FF_j)\big) \\ \nonumber
&+& \sum_{s=p+1}^{m}d_s\E\big(G(\tfrac{s}{n},\FF_{m+1-s})G(\tfrac{m+1}{n},\FF_j)\big) \label{New-Cov-2} 
	\end{eqnarray}
	($ m-p+1\leq j\leq m$).
    Define
    \begin{align*}
    \tilde \sigma_{i,j} & =\E(G(\tfrac{i}{n},\FF_i)G(\tfrac{j}{n},\FF_j))~,~~1\leq i\leq m-p~,~~1\leq j\leq m-p ,\\
     \tilde \sigma_{i,j}& =\E(G(\tfrac{i}{n},\FF_i)G(\tfrac{m+1}{n},\FF_j),~~
     1\leq i\leq  m-p~,~~m-p+1\leq j\leq {m+1},  \\
     \tilde \sigma_{i,j} & =\E(G(\tfrac{m+1}{n},\FF_i)G(\tfrac{j}{n},\FF_j))~,~~
     m-p+1\leq i\leq {m+1}~,~~  1\leq j\leq  m-p, \\
     \tilde \sigma_{i,j} &=\E(G(\tfrac{m+1}{n},\FF_i)G(\tfrac{m+1}{n},\FF_j))~,~~ m-p+1\leq i\leq { m+1}~,~~m-p+1\leq j\leq {m+1},
     \end{align*}
(note that $\tilde \sigma_{i,j}=\tilde \sigma_{j,i}$).
	These notations  and the equations  \eqref{New-Cov-1} and \eqref{New-Cov-2} show that  the $m$-dimensional
	 vector $\tilde {\bf a}_m=\big(0,..,0,a_p(\frac{m+1}{n}),a_{p-1}(\frac{m+1}{n}),...,a_1(\frac{m+1}{n})\big)^\top$ satisfies
	\begin{align*}
	\Sigma_m^{AR}\tilde{\bf a}_m=\bs \gamma_m^{AR},
	\end{align*}
	where the  $m\times m$ matrix  $\Sigma_m^{AR}$   and the $m$-dimensional vector $\bs \gamma_m^{AR}$ are defined by
	$\Sigma_m^{AR}=(\tilde \sigma_{i,j})_{1\leq i,j\leq m}$  and  $\gamma_m^{AR}=(\tilde \sigma_{m+1,1},...,\tilde \sigma_{m+1,m})^\top$, respectively.  On the other hand
	we have
	\begin{align*}
	\hat X_{m+1,n}^{\rm pred} &=\hat \mu^{1:m}(\tfrac{m}{n})+\sum_{s=1}^m \hat a_{m+1-s,n}(\mu(\tfrac{s}{n})-\hat \mu^{1:m}(\tfrac{s}{n}))+(\mathbf {\hat {a}}_m^*)^\top\mathbf{Z},\\
	~X_{m+1,n}&=\mu(\tfrac{m+1}{n})+\tilde{\mathbf{a}}_m^\top\tilde{\mathbf Z}+ {\sigma(\tfrac{m+1}{n})\varepsilon_{m+1}}.
	\end{align*}
	where  the $m$-dimensional vectors  ${ \mathbf{Z}}$ and $\tilde{ \mathbf{Z}}$ are given by
	 \begin{align*}
	 { \mathbf{Z}} &=\Big(G(\tfrac{1}{n},\FF_1),
	 ,...,
	 ,G(\tfrac{m+1-p}{n},\FF_{m+1-p}),G(\tfrac{m+2-p}{n},\FF_{m+2-p}),... , G(\tfrac{m}{n},\FF_m)\Big)^\top\\
		\tilde{\mathbf{Z}} & =\Big(G(\tfrac{1}{n},\FF_1),
		,...,
		G(\tfrac{m-p}{n},\FF_{m-p})
		,G(\tfrac{m+1}{n},\FF_{m+1-p}),G(\tfrac{m+1}{n},\FF_{m+2-p}),...G(\tfrac{m+1}{n},\FF_m)\Big)^\top.
\end{align*}	
(note that the first $m-p$ elements of the two vectors coincide). 
	Therefore we obtain  the following decomposition
	\begin{align*}
	&\hat X^{\rm Pred}_{m+1,n}-X_{m+1,n}=W_1+W_2- {\sigma(\tfrac{m+1}{n})}\varepsilon_{m+1},
		\end{align*}
		where
	\begin{align*}		
 W_1 & =\hat \mu^{1:m}(\tfrac{m}{n})-\mu(\tfrac{m+1}{n})+\sum_{s=1}^m \hat a_{m+1-s,n}\big(\mu(\tfrac{s}{n})-\hat \mu^{1:m}(\tfrac{s}{n})\big),\\
	W_2 & =(\mathbf {\hat {a}}_m^*)^\top\mathbf{Z}-\tilde{\mathbf{a}}_m^\top\tilde{\mathbf Z}:=W_{2,1}+W_{2,2},
	\\
	W_{2,1}& =(\mathbf {\hat {a}}_m^*)^\top(\mathbf Z-\tilde{\mathbf Z}), W_{2,2}=((\mathbf {\hat {a}}_m^*)^\top-\tilde{\mathbf a}_m^\top)\tilde {\mathbf Z}.
	\end{align*}
	It now  follows from  the proof of \eqref{boundamplus1} that $W_1=O_{\mathbb{P}}(r_n^\circ)$. To derive a similar estimate for  the term $W_{2,1}$ we
	 note that by \eqref{boundastarnorm} and \eqref{bound-astar}
	\begin{align*} 
	|\mathbf {\hat {a}_m}^*|=O_{\mathbb{P}}(1).
	\end{align*}
	Straightforward but tedious calculations  using condition (L2) yield that
	\begin{align*} 
	|\mathbf Z-\tilde{\mathbf Z}|=O_{\mathbb{P}}\Big(\frac{p^{3/2}}{n}\Big),
	\end{align*}
which leads to $W_{2,1}=O_{\mathbb{P}}(r_n)$. For estimation of   $W_{2,2}$,
note that a maximal inequality  shows
	\begin{align}\label{hol3}
| \mathbf{\tilde Z} |_{\infty}  = \max_{i=1}^m | \tilde  Z_i| 
=O_{\mathbb{P}}(n^{\frac{1}{q}})~.
	\end{align}
We will show below  that
	\begin{align}\label{tilde_a}
	|\tilde {\bf a}_m-\mathbf a_m^*|=O(\frac{p}{n}).
	\end{align}
which yields  with \eqref{bound-astar}
the estimate $|\mathbf {\hat {a}}_m^*-\tilde{\mathbf a}_m|=O_{\mathbb{P}}(r_n)$. Observing   \eqref{hol3} we have
	 $W_{2,2}=O_{\mathbb{P}}(n^{\frac{1}{q}}r_n)$, which completes the proof of part (b),  observing  the fact that $\varepsilon_{m+1}$
	 is identically distributed with $\varepsilon_{1}
$.	

\noindent
In order to show \eqref{tilde_a} we use   conditions (P2), (P3), will prove that
	\begin{align}\label{5.79}
	&\sup_{1\leq i,j\leq m+1}|\tilde \sigma_{i,j}-\sigma_{i,j}|=O\Big(\frac{2m+2- h_{m}(i) - h_{m}(j) }{n}\chi^{|i-j|}\Big),
		\end{align}
where
\begin{align*}
h_{m}(u) =(m+1)\mathbf{1}(1\leq u\leq m-p)+u\mathbf 1(m-p+1\leq u\leq m+1)~.
	\end{align*}
	To see this, we consider exemplarily the case that $i,j\in[m-p+1,m+1]$ - all other cases are treated in the same way.  Then
	\begin{align}
	\tilde \sigma_{i,j}-\sigma_{i,j}&=\E \big(G(\tfrac{m+1}{n},\FF_i)G(\tfrac{m+1}{n},\FF_j)\big)-\E \big(G(\tfrac{i}{n},\FF_i)G(\tfrac{j}{n},\FF_j)\big) = K_1+K_2, \label{hol4}
	\end{align}	
where $K_1$ and $K_2$ are defined by
		\begin{align}
	K_{1} &=\E\Big(G(\tfrac{m+1}{n},\FF_i)\Big(G(\tfrac{m+1}{n},\FF_j)-G(\tfrac{j}{n},\FF_j)\Big)\Big) 	\nonumber \\
	K_{2} &= \E\Big(G(\tfrac{j}{n},\FF_i)\Big(G(\tfrac{m+1}{n},\FF_i)-G(\tfrac{i}{n},\FF_i)\Big)\Big)  \nonumber
	\end{align}
 For  the investigation of $K_{1}$, we use
	the differentiability of the filter to obtain  
	\begin{align}\label{5.81}
	\Big|\E\Big(G(\tfrac{m+1}{n},\FF_i)\big(G(\tfrac{m+1}{n},\FF_j)-G(\tfrac{j}{n},\FF_j)\big)\Big)\Big|&
	&\leq \int_{\tfrac{j}{n}}^{\tfrac{m+1}{n}}\Big|\E\Big(G(\tfrac{m+1}{n},\FF_i)\dot G(u,\FF_j)\Big)\Big|du
	\end{align}
Observing assumption (P2), (P3) and by the argument of proving \eqref{small-cov},  it follows
	\begin{align}\label{5.82}
	\big|\E\big(G(\frac{m+1}{n},\FF_i)\dot G(u,\FF_j)\big)\big|=O(\chi^{|i-j|})
	\end{align}
	(uniformly with respect  to $u\in [0,1]$).
Combining the estimates  \eqref{5.81} and \eqref{5.82} yields
	\begin{align*}
	|K_{1}|=O\Big(\frac{m+1-j}{n}\chi^{|i-j|}\Big).
	\end{align*}
	Similarly it follows that
	$|K_{2}|=O(\frac{m+1-i}{n}\chi^{|i-j|})$.
	These  bounds and \eqref{hol4}  yield
	\begin{align*}
		\sup_{m-p\leq i,j\leq m+1}|\tilde \sigma_{i,j}-\sigma_{i,j}|=O\Big(\frac{2m+2-i-j}{n}\chi^{|i-j|}\Big),
		\end{align*}
		which shows that \eqref{5.79} holds uniformly for $i,j\in [m-p,m+1]$. Similar and simpler arguments yield that \eqref{5.79} holds uniformly for the  other choices of $i,j$.

		Next, observe that $\Sigma_m^{AR}$ is an $m\times m$ symmetric matrix, and so is $\Sigma_m^{AR}-\Sigma_{n,m}$. By similar arguments as given in the proof of
		 Theorem \ref{thm2}, it follows that
		\begin{align*}
	\rho(\Sigma_m^{AR}-\Sigma_{n,m})\leq \max_{1\leq i\leq n} \sum_{j=1}^n|\sigma_{i,j}-\tilde \sigma_{i,j}|=O(\frac{p}{n}),
		\end{align*}
		where the last inequality is a consequence from  \eqref{5.79}. This inequality and  assumption
		(E1) imply that $\Sigma_m^{AR}$ is positive definite if  $n$ is sufficiently large. Consequently,
			\begin{align*}
		\tilde{\bf a}_m=(\Sigma_m^{AR})^{-1}\bs \gamma_m^{AR}.
			\end{align*}
			and by similar arguments as  given in the  proof of Corollary \ref{Corol2} we obtain that
			\begin{align}\label{Final_1}
		\rho((\Sigma_m^{AR})^{-1}-\Sigma^{-1}_{n,m}) & =O(\frac{p}{n}) , \\
		\label{Final_2}|\bs \gamma_m^{AR}-\bs \gamma_m|
		&=O(\frac{p}{n})
		\end{align}
		Now \eqref{tilde_a} follows from \eqref{star_a} \eqref{Final_1}, \eqref{Final_2}, which completes the proof.\hfill $\Box$

\bigskip\bigskip

\noindent {\bf Acknowledgements}
This work has been supported in part by the
Collaborative Research Center ``Statistical modeling of nonlinear
dynamic processes'' (SFB 823, Project A1, C1) of the German Research Foundation (DFG) and  NSFC Young program (No.11901337).

{
\setlength{\bibsep}{2pt}

	\bibliography{lit}
}

\end{document}